\documentclass[runningheads]{llncs} \usepackage{graphicx}
\usepackage{comment} \usepackage{url}
\usepackage{mathtools} \usepackage{bussproofs} \usepackage{pdftexcmds}
\usepackage{stmaryrd} \usepackage{amssymb} \usepackage{color}
\usepackage{xspace} \usepackage{nicefrac}
\usepackage{etoolbox} \usepackage{marvosym}
\includecomment{techreport}
\excludecomment{cameraready}

\makeatletter
\RequirePackage[bookmarks,unicode,colorlinks=true]{hyperref}%
   \def\@citecolor{blue}%
   \def\@urlcolor{blue}%
   \def\@linkcolor{blue}%

\def\orcidID#1{\smash{\href{http://orcid.org/#1}{\protect\raisebox{-1.25pt}{\protect\includegraphics{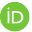}}}}}
\makeatother

\makeatletter
\patchcmd{\@verbatim}
  {\verbatim@font}
  {\verbatim@font\small}
  {}{}
\makeatother

\newcommand{\tosql}[1]{(#1)^{\mathsf{sql}}}

\newcommand{\QueL}{\textsc{Que}$\Lambda$\xspace}
\newcommand{\QueLG}{\textsc{Que}$\Lambda_G$\xspace}
\newcommand{\SQUR}{\textsc{Squr}\xspace}
\DeclareMathOperator{\FV}{FV}
\DeclareMathOperator{\dom}{dom}
\newcommand{\sem}[1]{\left\llbracket {#1} \right\rrbracket}

\newcommand{\vect}[1]{\overrightarrow{#1}}

\newcommand{\red}[0]{\mathrel{\leadsto}}
\newcommand{\monus}{\stackrel{.}{-}}

\def\orelse{\mathbin{~|~}}

\newcommand{\subst}[2]{\left[\nicefrac{#2}{#1}\right]}
\newcommand{\bbI}[0]{\mathbb{I}}

\newcommand{\cG}[0]{\mathcal{G}}

\newcommand{\cS}[0]{\mathcal{S}}

\def\rVS{\noalign{\vskip 3ex plus.4ex minus.4ex}}

\newcommand{\NRC}[0]{\ensuremath{\mathcal{NRC}}\xspace}
\newcommand{\NRCl}[0]{\ensuremath{\NRC_\lambda\xspace}}
\newcommand{\NRClsb}[0]{\ensuremath{\NRC_\lambda(\mathit{Set,Bag})}\xspace}
\newcommand{\NRCg}[0]{\ensuremath{\NRC_\cG}\xspace}
\newcommand{\BQL}[0]{\ensuremath{\mathcal{BQL}}\xspace}
\def\boolty{\mathbf{B}}
\def\emptyoset{\emptyset}
\def\emptymset{\mho}
\def\ocup{\cup}
\def\comprehension{\bigcup\setlit}
\def\mcomprehension{\biguplus\msetlit}
\newcommand{\setlit}[1]{\{{#1}\}}
\newcommand{\msetlit}[1]{\Lbag{#1}\Rbag}

\newcommand{\isempty}{\mathit{empty}}
\newcommand{\norm}{\mathit{norm}}
\newcommand{\tuple}[1]{\langle{#1}\rangle}

\def\distinct{\delta}
\def\promote{\iota}
\def\plwhere{\mathbf{where}}
\def\setwhere{\mathbf{where}_{\mathsf{set}}}
\def\bagwhere{\mathbf{where}_{\mathsf{bag}}}

\def\setempty{\mathbf{empty}_{\mathsf{set}}}
\def\bagempty{\mathbf{empty}_{\mathsf{bag}}}

\def\kwempty{\mathbf{empty}}

\def\kwwhere{\mathbf{where}}

\def\kwdo{\mathbf{do}}
\def\kwindex{\mathit{index}}
\DeclareSymbolFont{symbolsC}{U}{txsyc}{m}{n}
\DeclareMathSymbol{\strictfun}{\mathrel}{symbolsC}{74}
\newcommand{\graph}[2]{\mathcal{G}(#1;#2)}
\newcommand{\setgraph}[2]{\mathcal{G}^{\mathsf{set}}(#1;#2)}
\newcommand{\baggraph}[2]{\mathcal{G}^{\mathsf{bag}}(#1;#2)}
\def\gapp{\varoast} 
\makeatletter
\def\fdsy@scale{1}
\newcommand\fdsy@mweight@normal{Book}
\newcommand\fdsy@mweight@small{Book}
\newcommand\fdsy@bweight@normal{Medium}
\newcommand\fdsy@bweight@small{Medium}
\DeclareFontFamily{U}{FdSymbolC}{}
\DeclareFontShape{U}{FdSymbolC}{m}{n}{
    <-7.1> s * [\fdsy@scale] FdSymbolC-\fdsy@mweight@small
    <7.1-> s * [\fdsy@scale] FdSymbolC-\fdsy@mweight@normal
}{}
\DeclareFontShape{U}{FdSymbolC}{b}{n}{
    <-7.1> s * [\fdsy@scale] FdSymbolC-\fdsy@bweight@small
    <7.1-> s * [\fdsy@scale] FdSymbolC-\fdsy@bweight@normal
}{}
\makeatother

\DeclareSymbolFont{fdarrows}{U}{FdSymbolC}{m}{n}
\SetSymbolFont{fdarrows}{bold}{U}{FdSymbolC}{m}{n}

\DeclareMathSymbol{\leftpitchfork}{\mathrel}{fdarrows}{118}
\def\flat{\breve}

\def\shreds{\Mapsto}
\newcommand{\banana}[1]{\llparenthesis #1 \rrparenthesis}
\newcommand{\stitch}[2]{\banana{#2}#1}

\newcommand{\defun}[1]{\left\lfloor {#1} \right\rfloor}

\def\kwsel{{\color{blue}\mathtt{SELECT}}}
\def\kwapply{{\color{blue}\mathtt{APPLY}}}
\def\kwdist{{\color{blue}\mathtt{DISTINCT}}}
\def\kwfrom{{\color{blue}\mathtt{FROM}}}
\def\kwsqlwhere{{\color{blue}\mathtt{WHERE}}}
\def\kwas{{\color{blue}\mathtt{AS}}}
\def\kwunion{{\color{blue}\mathtt{UNION}}}

\def\kwexcept{{\color{blue}\mathtt{EXCEPT}}}
\def\kwlateral{{\color{blue}\mathtt{LATERAL}}}

\def\kwall{{\color{blue}\mathtt{ALL}}}

\def\kwnot{{\color{blue}\mathtt{NOT}}}

\def\kwtrue{{\mathsf{true}}}
\def\kwfalse{{\mathsf{false}}}
\def\kwex{{\color{blue}\mathtt{EXISTS}}}

\def\kwrownum{{\color{blue}\mathtt{ROW\_NUMBER}}}
\newcommand{\NF}{M}
\newcommand{\NNF}{N}
\newcommand{\atNF}{\NF_\mathsf{atom}}
\newcommand{\boolNF}{\NF_\mathsf{bool}}
\newcommand{\funNF}{\NF_\mathsf{fun}}
\newcommand{\tupNF}{\NF_\mathsf{tup}}
\newcommand{\setNF}{\NF_\mathsf{set}}
\newcommand{\bagNF}{\NF_\mathsf{bag}}

\newcommand{\boolNNF}[1]{\NNF_\mathsf{bool}^\mathsf{#1}}

\newcommand{\setNNF}[1]{\NNF_\mathsf{set}^\mathsf{#1}}
\newcommand{\bagNNF}[1]{\NNF_\mathsf{bag}^\mathsf{#1}}

\definecolor{dred}{RGB}{127,0,0}
\definecolor{dgreen}{RGB}{0,127,0}
\definecolor{wrbk}{rgb}{0.8,0.9,1}
\definecolor{jcbk}{rgb}{0.8,1,0.8}
\definecolor{wrtxt}{rgb}{0,0.5,0.8}
\definecolor{jctxt}{rgb}{0,0.6,0}

\newcommand{\wrnote}[1]{\colorbox{wrbk}{\parbox{0.9\textwidth}{\textcolor{wrtxt}{WR:} #1}}}
\newcommand{\jcnote}[1]{\colorbox{jcbk}{\parbox{0.9\textwidth}{\textcolor{jctxt}{JRC:} #1}}}

\begin{document}

%
\title{Query Lifting}
\subtitle{Language-integrated query 
\\
for heterogeneous nested collections}
\author{Wilmer Ricciotti\inst{1} (\Letter) \orcidID{0000-0002-2361-8538} \and
James Cheney\inst{1,2}\orcidID{0000-0002-1307-9286}}
\authorrunning{W. Ricciotti and J. Cheney}
%
\institute{
  Laboratory for Foundations of Computer Science\\
  University of Edinburgh, Edinburgh, United Kingdom\\
  \email{research@wilmer-ricciotti.net}\\
  \email{jcheney@inf.ed.ac.uk}
\and
The Alan Turing Institute, London, United Kingdom
}

\maketitle
\begin{abstract}
  Language-integrated query based on comprehension syntax is a
  powerful technique for safe database programming, and provides a
  basis for advanced techniques such as \emph{query shredding} or
  \emph{query flattening} that allow efficient programming with
  complex nested collections. However, the foundations of these
  techniques are lacking: although SQL, the most widely-used database
  query language, supports \emph{heterogeneous} queries that mix set
  and multiset semantics, these important capabilities are not
  supported by known correctness results or implementations that
  assume \emph{homogeneous} collections.  In this paper we study
  language-integrated query for a heterogeneous query language
  \NRClsb that combines set and multiset constructs.  We show how to
  normalize and translate queries to SQL, and develop a novel approach
  to querying heterogeneous nested collections, based on the insight
  that ``local'' query subexpressions that calculate nested
  subcollections can be ``lifted'' to the top level analogously to
  lambda-lifting for local function definitions.

\end{abstract}

\section{Introduction}

Since the rise of relational databases as important software
components in the 1980s, it has been widely appreciated that database
programming is hard~\cite{copeland-maier:1984}.  Databases offer
efficient access to flat tabular data using declarative SQL queries, a
computational model very different from that of most general-purpose
languages.  To get the best performance from the database, programmers
typically need to formulate important parts of their program's logic
as queries, thus effectively programming in two languages: their usual
general-purpose language (e.g. Java, Python, Scala) and SQL, with the
latter query code typically constructed as unchecked, dynamic strings.
Programming in two languages is more than twice as difficult as
programming in one language~\cite{lindley10esop}.  The result is a hybrid
programming model where important parts of the program's functionality
are not statically checked and may lead to run-time failures, or
worse, vulnerabilities such as SQL injection attacks.  This
undesirable state of affairs was recognized by Copeland and
Maier~\cite{copeland-maier:1984} who coined the term \emph{impedance
  mismatch} for it.

Though higher-level wrapper libraries and tools such as
\emph{object-relational mappings} (ORM) can help ameliorate the
impedance mismatch, they often come at a price of performance and lack
of transparency, as high-level operations on in-memory objects
representing database data are not always mapped efficiently to
queries~\cite{russell08queue}.  An alternative approach, which has
almost as long a history as the impedance mismatch problem itself, is
to elevate queries in the host language from unchecked strings to a
typed, domain-specific sublanguage, whose interactions with the rest of
the program can be checked and which can be mapped to database queries
safely while providing strong guarantees.  This approach is nowadays
typically called \emph{language-integrated query} following
Microsoft's successful LINQ extensions to .NET languages such as C\#
and F\#~\cite{meijer:sigmod,Syme06}.  It is ultimately based on
Trinder and Wadler's insight that database queries can be modeled by a
form of monadic comprehension
syntax~\cite{trinder-wadler:comprehensions}.

Comprehension-based query languages were
placed on strong foundations in the database
community in the 1990s~\cite{buneman+:comprehensions,BNTW95,ParedaensG92,wong96jcss,LW97}.
A key insight due to Paredaens and van Gucht~\cite{ParedaensG92} is
that although comprehension-based queries can manipulate nested
collections, any expression whose input and output are \emph{flat}
collections (i.e. tables of records without other collections nested
inside field values) can always be translated to an equivalent query
only using flat relations (i.e. can be expressed in an SQL-like
language).  Wong~\cite{wong96jcss} subsequently generalized this
result and gave a constructive proof, in which the translation from
nested to flat queries is accomplished through a strongly normalizing
rewriting system.

Wong's work has informed a number of successful implementations, such
as the influential Kleisli system~\cite{Won00} for biomedical data
integration, and the Links programming language~\cite{CLWY06}.
Although the implementation of LINQ in C\# and F\# was not directly
based on normalization, Cheney et al.~\cite{cheney13icfp} showed that normalization can be performed
as a pre-processing step to improve both reliability and performance
of queries, and guarantee that a well-formed query expression
evaluates to (at most) one equivalent SQL expression at run
time.

Comprehension-based language-integrated query also forms the basis for
libraries such as Quill for Scala~\cite{quill} and Database-Supported
Haskell~\cite{dsh}.  Most recently, language-integrated query has been
extended further to support efficient execution of queries that
construct \emph{nested}
results~\cite{GrustMRS09,cheney14sigmod,dsh,SIGMOD2015UlrichG}, by
translating such queries to a bounded number of flat queries.  This
technique, currently implemented in Links and DSH, has several
benefits: for example to implement \emph{provenance-tracking}
efficiently in queries~\cite{fehrenbach18scp,stolarek18programming}.
Fowler et al.~\cite{fowler20ijdc} showed that in some cases, Links's
support for nested query results decreased both the number of
queries issued and the total query evaluation time by an order of
magnitude or more compared to a Java database application.
Unfortunately, there is still a gap between the theory and practice of
lan\-guage-integrated query.  Widely-used and practically important SQL
features that mix set and multiset collections, such as duplicate
elimination, are supported by some implementations, but without
guarantees regarding correctness or reliability.  So far, such results
have only been proved for special
cases~\cite{cheney13icfp,cheney14sigmod}, typically for
\emph{homogeneous} queries operating on one uniform collection type.
For example, in Links, queries have multiset semantics and cannot use
duplicate elimination or set-valued operations.  To the best of our
knowledge the questions of how to correctly translate flat or nested
\emph{heterogeneous} queries to SQL are open problems.

In this paper, we solve both open problems.  We study a
heterogeneous query language \NRClsb, which was introduced and
studied in our recent work~\cite{ricciotti19dbpl}.  We have previously
extended the key results on query normalization to
\NRClsb~\cite{ricciotti20fscd}, but unlike the homogeneous case, the
resulting normal forms do not directly correspond to SQL.  In this
paper, we first show how flat \NRClsb queries can be translated to
SQL, and we then develop a new approach for evaluating queries over
nested heterogeneous collections.  The key (and, to us at least,
surprising) insight is to recognize that these two subproblems are
really just different facets of one problem.  That is, when
translating flat \NRClsb queries to SQL, the main obstacle is how to
deal with query expressions that depend on local variables; when
translating nested \NRClsb queries to equivalent flat ones, the main
obstacle is also how to deal with query expressions that depend on
local variables.  We solve this problem by observing that such query
subexpressions can be \emph{lifted}, analogously to
\emph{lambda-lifting} of local function definitions in functional
programming~\cite{johnsson85fpca}, by abstracting over their free
variables.  Differently to lambda-lifting, however, we lift such
expressions by converting them to \emph{tabular functions}, or
\emph{graphs}, which can be calculated using database query
constructs.

The remainder of this paper presents our contributions as follows:
\begin{itemize}
\item In section~\ref{sec:overview} we review the most relevant prior
  work and present our approach at a high, and we hope accessible,
  level.
\item In sections~\ref{sec:background} and~\ref{sec:graph} we present the core languages \NRClsb
  and \NRCg which will be used in the rest of the paper.
\item Section~\ref{sec:delat} presents our results on translation of
  flat \NRClsb queries to SQL, via \NRCg.
\item Section~\ref{sec:shredding} presents our results on translation
  of \NRClsb queries that construct nested results to a bounded number
  of flat \NRCg queries.
\item Sections~\ref{sec:related} and~\ref{sec:concl} discuss related
  work and conclude.
\end{itemize}

\section{Overview}\label{sec:overview}

In this section we sketch our approach.  We use Links
syntax~\cite{CLWY06}, which differs in superficial respects from the
core calculus in the rest of the paper but is more readable.  We rely
without further comment on existing capabilities of
language-integrated query in Links, which are described
elsewhere~\cite{Cooper09,lindley12tldi,cheney14sigmod}.  Suppose,
hypothetically, we are interested in certain presidential candidates
and prescription drugs they may be taking\footnote{For example, to see whether drug
  interactions might explain erratic behavior such as rage
  tweeting, creeping authoritarianism, or creepiness more generally.}.  In Links,
an expression querying a small database of presidential candidates and
their drug prescriptions can be written as follows:
\begin{verbatim}
Q0 = for (c <- Cand, p <- Pres, d <- Drug) 
     where (c.cid == p.cid && p.did == d.did) 
     [(name=c.name,drug=d.drug)]
\end{verbatim}
Some (totally fictitious and not legally actionable) example data is
shown in Figure~\ref{fig:data1}; note that the prescriptions table
$Pres$ is a multiset containing duplicate entries.  Executing
this query in Links results in the following SQL query:
\begin{verbatim}
SELECT c.name, d.drug 
FROM Cand c, Pres p, Drug d  
WHERE c.cid = p.cid AND p.did = d.did
\end{verbatim}
In Links, query results from the database are mapped back to list
values nondeterministically, and the result of the above query $Q_0$
will be a list containing two copies of the tuple $(\mathtt{DJT},\mathtt{adderall})$ and
one copy of each of the tuples $(\mathtt{DJT},\mathtt{hydrochloroquine})$ and
$(\mathtt{JRB},\mathtt{caffeine})$.  If we are just interested in which candidates take
which drugs and not how many times each drug was taken, we want to
remove these duplicates.  This can be accomplished in a basic SQL
query using the $\kwdist$ keyword after $\kwsel$.  Currently, in Links
there is no way to generate queries involving $\kwdist$, and this
duplicate elimination can only be performed in-memory.  While this is
not hard to do when the duplicate elimination happens at the end of
the query, it is not as clear how to handle deduplication operations
correctly in arbitrary places inside queries.  Furthermore, SQL has several
other operations that can have either set or multiset semantics such
as $\kwunion$ and $\kwexcept$: how should they be handled?

To study this problem we introduced a core calculus
$\NRClsb$~\cite{ricciotti19dbpl} (reviewed in the next section) in
which there are two collection types, sets and multisets (or
\emph{bags}); duplicate elimination maps a multiset to a set with the
same elements, and \emph{promotion} maps a set to the least multiset
with the same elements.

\begin{figure}[tb]\centering
  \includegraphics[scale=0.3]{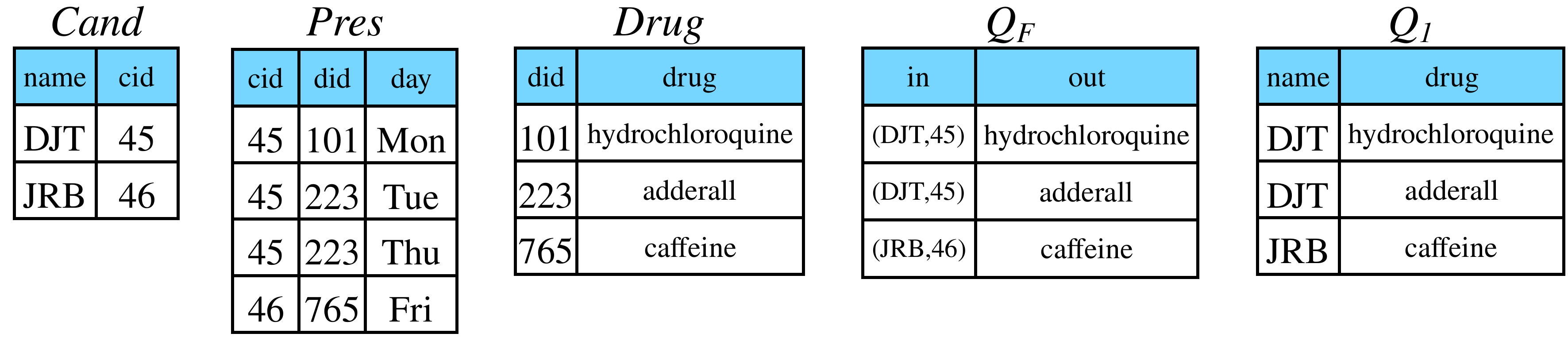} 
  \caption{Input tables $Cand,Pres,Drug$, intermediate result of $Q_F$
    and result of $Q_1$.}\label{fig:data1}
\end{figure}

We considered, but were not previously able to solve, two problems in
the context of $\NRClsb$ which are addressed in this paper.  First,
the fundamental results regarding normalization and translation to SQL
have been studied only for \emph{homogeneous} query languages with
collections consisting of either sets, bags, or lists.  We recently
extended the normalization results to
$\NRClsb$~\cite{ricciotti20fscd}, but the resulting normal forms do
not correspond directly to SQL queries if operations such as
deduplication, promotion, or bag difference are present.  Second,
query expressions that construct nested collections cannot be
translated directly to SQL and can be very expensive to execute
in-memory using nested loops, leading to the \emph{$N+1$ query
  problem} (or \emph{query avalanche problem}~\cite{grust10vldb}) in
which one query is performed for the outer loop and then another $N$
queries are performed, one per iteration of the inner loop.  Some
techniques have been developed for translating nested queries to a
fixed number of flat queries, but to date they either handle only
homogeneous set or bag collections~\cite{Bussche01,cheney14sigmod}, or
lack detailed correctness proofs~\cite{grust10vldb,ulrich19phd}.

Regarding the first problem, the closest work in this respect is by
Libkin and Wong~\cite{LW97}, who studied and related the
expressiveness of comprehension-based homogeneous set and bag query
languages but did not consider their heterogeneous combination or
translation to SQL.  The following query illustrates the fundamental
obstacle:
\begin{verbatim}
Q1 = for (c <- Cand)
     for (d <- dedup(for (p <- Pres, d <- Drug) 
                     where (c.cid == p.cid && p.did == d.did) 
                     [d.drug]))
     [(name=c.name, drug=d)]
\end{verbatim}
This query is similar to $Q_0$, but eliminates duplicates among the
drugs for each candidate.  The query contains a duplicate elimination
operation (\verb|dedup|) applied to another query subexpression that
refers to $c$, which is introduced in an earlier generator.  This is
not directly supported in classic SQL: by default the subqueries in
$\kwfrom$ clauses cannot refer to tuple variables introduced by
earlier parts of the $\kwfrom$ clause.  In fact, this query is
expressible in SQL:1999 using the $\kwlateral$ keyword, which does
allow such sideways information-passing:
\begin{verbatim}
SELECT c.name,d.drug
FROM Cand c, LATERAL (SELECT DISTINCT d.drug 
                      FROM Pres p, Drug d 
                      WHERE p.cid = c.cid AND p.did = d.did) d
\end{verbatim}
(Without the $\kwlateral$ keyword, this query is not well-formed SQL.)
However, such queries have only recently become widely supported, so
are not available on legacy databases, and even when supported, are
not typically optimized effectively; for example PostgreSQL will
evaluate it as a nested loop, with quadratic complexity or worse.

Regarding the second problem, Van den Bussche~\cite{Bussche01} showed that
any query returning nested set collections can be simulated by $n$
flat queries, where $n$ is the number of occurrences of the set
collection type in the result.  However, this translation has not been
used as the basis for a practical system to our knowledge, and does
not respect multiset semantics.  Cheney et
al.~\cite{cheney14sigmod} provided an analogous \emph{shredding}
translation for nested multiset queries, but translated to a richer
target language (including SQL:1999 features such as $\kwrownum$) and
did not handle operations such as multiset difference or duplicate
elimination.  Thus, neither approach handles the full expressiveness
of a heterogeneous query language over bags and sets.  The following
query illustrates the fundamental obstacle:
\begin{verbatim}
Q2 = for (x <- Cand)
     [(name=x.name, drugs=dedup(for (p <- Pres, d <- Drug) 
                             where (x.cid == p.cid and p.did == d.did) 
                             [d.drug]))]
\end{verbatim}
Much like $Q_1$, $Q_2$ builds a multiset of pairs $(name,drugs)$ but here
$drugs$ is a \emph{set} of all of the drugs taken by candidate $name$.  Such a query is, of course, not even
syntactically expressible in SQL because it returns a nested
collection; it is not expressible in previous work on nested query
evaluation either, because the result is a multiset of records, one
component of which is a set.

We will now illustrate how to translate $Q_1$ to a plain SQL query
(not using $\kwlateral$), and how to translate $Q_2$ to two flat
queries such that the nested result can be constructed easily from
their flat results.  First, note that we can rewrite both queries as
follows, introducing an abbreviation $F(x)$ for a query subexpression
parameterized by $x$:
\begin{verbatim}
F(x) = for (p <- Pres, d <- Drug) 
       where (x.cid == p.cid and p.did == d.did) 
       [d.drug]
Q1   = for (c <- Cand) for (d <- dedup(F(c))) [(name=c.name, drug=d)]
Q2   = for (c <- Cand) [(name=c.name, drugs=dedup(F(c)))]
\end{verbatim}
Next, observe that the set of all possible values for $x$ appearing in
some call to $F(x)$ is finite, and can even be computed by a query.
Therefore, we can write a \emph{closed} query $Q_F$ that builds a
lookup table that calculates the \emph{graph} of $F$ (or at least, as
much of it as is needed to evaluate the queries) as follows:
\begin{verbatim}
Q_F = dedup(for (x <- Cand, y <- F(x)) [(in=x,out=y))]
\end{verbatim}
Notice that the use of deduplication here is really essential to
define $Q_F$ correctly: if we did not deduplicate then there would be repeated
tuples in $Q_F$, leading to incorrect results
later.  If we inline and simplify $F(x)$ in the above query, we get
the following:
\begin{verbatim}
Q_F' = dedup(for (x <- Cand, y <- Pres, z <- Drug) 
             where (x.cid == y.cid && y.did = z.did) 
             [(in=x,out=z.drug)])
\end{verbatim}
Finally we may replace the call to $F(x)$ in $Q_1$ with a
lookup to $Q_F'$, as follows:
\begin{verbatim}
Q1' = for (c <- Cand, f <- Q_F') where (c == f.in) 
      [(name=c.name, drug=f.out)]
\end{verbatim}
This expression may now be translated
directly to SQL, because the
argument to \verb|dedup| is now closed:
\begin{verbatim}
SELECT c.name,f.drug
FROM Cand c, (SELECT DISTINCT x.name,x.cid,z.drug 
              FROM Cand x, Pres y, Drug z
              WHERE x.cid = y.cid AND y.did = z.did) f
WHERE c.cid = f.cid AND c.name = f.name
\end{verbatim}
Although this query looks a bit more complex than the one given
earlier using $\kwlateral$, it can be optimized more effectively, for
example PostgreSQL generates a query plan that uses a hash join,
giving quasi-linear complexity.  

\begin{figure}[tb]\centering
  \includegraphics[scale=0.3]{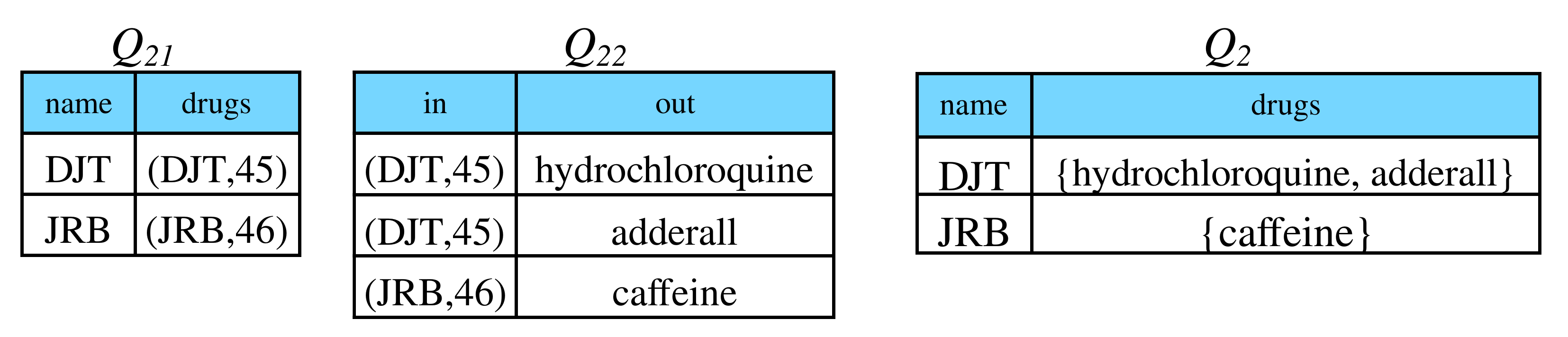} 
  \caption{Intermediate results of $Q_{21},Q_{22}$
    and result of $Q_2$.}\label{fig:data2}
\end{figure}
 
On the other hand, to deal with $Q_2$, we refactor it into two closed,
flat queries $Q_{21},Q_{22}$ and an expression $Q_2'$ that builds the
nested result from their flat results (illustrated in
Figure~\ref{fig:data2}):
\begin{verbatim}
Q_21 = for (x <- Cand) [(name=x.name, drugs=x)]
Q_22 = Q_F
Q2'  = for (x <- Q21) 
       [(name=x.name, 
         drugs=for (y <- Q_22) where(x.drugs == y.in) [y.out])]
\end{verbatim}
Notice that in $Q_{21}$ we replaced the call to $F$ with the argument
$x$, while $Q_{22}$ is just  $Q_F$ again.
The final expression $Q_2'$ builds the nested result (in the host
language's memory) by traversing $Q_{21}$ and computing the set value
of each $cs$ field by looking up the appropriate values from $Q_{22}$.
Thus, the original query result can be computed by first evaluating
$Q_{21}$ and $Q_{22}$ on the database, and then evaluating
the final \emph{stitching} query expression in-memory.  (In practice,
as discussed in Cheney et al.~\cite{cheney14sigmod}, it is important
for performance to use a more sophisticated stitching algorithm than
the above naive nested loop, but in this paper we are primarily
concerned with the correctness of the transformation.)

The above examples are a bit simplistic, but illustrate the key idea
of \emph{query lifting}.  In the rest of this paper we place this
approach on a solid foundation, and (partially inspired by Gibbons et
al.~\cite{gibbons18icfp}), to help clarify the reasoning we extend the
calculus with a type of \emph{tabulated functions} or \emph{graphs}
$\vect{\sigma} \strictfun \setlit{\tau}$, with \emph{graph
  abstraction} introduction form $\graph{-}{-}$ and \emph{graph
  application} $M \gapp \tuple{\vect{x}}$. In our running example we
could define $Q_F = \graph{x \gets R}{F(x)}$, and we would use the
application operation $M\gapp \tuple{\vect{x}}$ to extract the set of
elements corresponding to $x$ in $Q_F$.  We will also consider tabular
functions that return multisets rather than sets, in order to deal
with queries that return nested multisets.

\section{Background}\label{sec:background}
We recap the main points from~\cite{ricciotti19dbpl}, which introduced
a calculus 
\linebreak
\NRClsb with the following syntax:

\begin{tabular}{l@{\hspace{1em}}rcl}
  \textbf{Types} & $\sigma, \tau$ & $::=$ & $b \orelse \tuple{\vect{\ell : \sigma}} \orelse \setlit{\sigma} \orelse \msetlit{\sigma} 
    \orelse \sigma \to \tau$
  \\
  \textbf{Terms} & $M, N$ & $::=$ & $x \orelse t \orelse c(\vect{M}) \orelse \tuple{\vect{\ell = M}} \orelse M.\ell \orelse \lambda x.M \orelse M~N$ \\
  & & $\orelse$ & $\emptyset \orelse \setlit{M} \orelse M \cup N \orelse \bigcup\setlit{M | \Theta} $\\
  & & $\orelse$ & $\emptymset \orelse \msetlit{M} \orelse M \uplus N
                  \orelse M - N
                  \orelse \biguplus\msetlit{M | \Theta}$\\
  & & $\orelse$ & $\distinct M \orelse \promote M \orelse M~\setwhere~N \orelse M~\bagwhere~N$ \\
  & & $\orelse$ & $\setempty(M) \orelse \bagempty(M)$ \\
  \textbf{Generators} & $\Theta$ & $::=$ & $\vect{x \gets M}$
\end{tabular}

We distinguish between (local) variables $x$ and (global) table names
$t$, and assume standard primitive types $b$ and primitive operations 
$c(\vect{M})$
including respectively Booleans $\boolty$ and equality at every base type.  The syntax for records and
record projection $ \tuple{\vect{\ell = M}}, M.\ell$, and for
lambda-abstraction and application $\lambda x.M, M~N$ is
standard; as usual, let-binding is definable. Set operations include
empty set $\emptyset$, singleton construction $\setlit{M}$, union
$M \cup N$, one-armed conditional $M~\setwhere~N$, emptiness test
$\setempty(M)$, and comprehension $\bigcup\setlit{M\mid \Theta}$,
where $\Theta$ is a sequence of generators $x \leftarrow M$.
Similarly, multiset operations include empty bag $\emptymset$,
singleton $\msetlit{M}$, bag union $M \uplus N$, bag difference
$M - N$, conditional $M~\bagwhere~N$, emptiness test $\bagempty(M)$. The syntax is completed by duplicate elimination $\distinct M$ (converting a bag $M$ into a set with the same object type) and promotion $\promote M$ (which produces the bag containing all the elements of the set $M$, with multiplicity 1). 

The one-way conditional operations $M~\setwhere~N$ and $M~\bagwhere~N$
evaluate Boolean test $N$, and return collection $M$ if $N$ is true,
otherwise the empty set/bag; two-way conditionals can supported
without problems.  Other set operations, such as intersection,
membership, subset, and equality are also definable, as are bag
operations such as intersection~\cite{BNTW95,LW97}.  Also, we may
define $\bagempty(M)$ as $\setempty(\distinct(M))$ and
$M~ \setwhere~ N$ as $\distinct(\promote(M)~ \bagwhere~ N)$, but we
prefer to include these constructs as primitives for
symmetry. Generally, we will allow ourselves to write $M~\kwwhere~N$
and $\kwempty(M)$ without subscripts if the collection kind of these
operations is irrelevant or made clear by the context.  We freely use
syntax for unlabeled tuples $\tuple{\vect{M}}, M.i$ and tuple types
$\vect{\sigma}$ and consider them to be syntactic sugar for labeled
records.

The typing rules for the calculus are standard and provided 
\begin{techreport}
in an appendix.
\end{techreport}
\begin{cameraready}
in the full version of this paper~\cite{ricciotti21lifting}.
\end{cameraready}
For the purposes of this discussion, we will highlight two features of the type system. The first is that the calculus used here differs from our previous work by using constants and table names, whose types are described by a fixed signature $\Sigma$:

\begin{center}
    \AxiomC{$\Sigma(c) = \vect{b} \to b$}
    \AxiomC{$(\Gamma \vdash M_i : \sigma_i)_{i=1,\ldots,n}$}
    \BinaryInfC{$\Gamma \vdash c(\vect{M}) : \tau$}
    \DisplayProof
  \hspace{1cm}
    \AxiomC{$\Sigma(t) = \vect{\ell : b}$}
    \UnaryInfC{$\Gamma \vdash t : \msetlit{\tuple{\vect{\ell : b}}}$}
    \DisplayProof
\end{center}

As usual, a typing judgment $\Gamma \vdash M : \sigma$ states that a term $M$ is well-typed of type $\sigma$, assuming that its free variables have the types declared in the typing context $\Gamma = x_1 : \sigma_1, \ldots, x_k : \sigma_k$.
For the two rules above, note in particular that the primitive functions $c$ can only take
inputs of base type and produce results at base type, and table
constants $t$ are always multisets of records where the fields are of
base type.  We refer to a type of the form $\tuple{\vect{\ell:b}}$ as
\emph{flat}; if $\sigma$ is flat, we refer to $\setlit{\sigma}$ and $\msetlit{\sigma}$ as \emph{flat collection types}.

The second is that our type system uses an approach à la Church,
meaning that
variable abstractions (in lambdas/comprehensions), empty sets and empty bags are annotated with their type in order to ensure the uniqueness of typing.

\begin{lemma}
In $\NRClsb$, if $\Gamma \vdash M : \sigma$ and $\Gamma \vdash M : \tau$, then $\sigma = \tau$.
\end{lemma}

In the context of a larger language implementation, most of these type
annotations can be elided and inferred by type inference. We have
chosen to dispense with these details in the main body of this paper
to avoid unnecessary syntactic cluttering.

We will use a largely standard denotational semantics for \NRClsb, in
which sets and multisets are modeled as finitely-supported functions
from their element types to Boolean values $\{0,1\}$ or natural
numbers respectively.  This approach follows the so-called
$K$-relation semantics for queries~\cite{green07pods,foster08pods} as
used for example in the HoTTSQL formalization~\cite{Chu17}.  The full
typing rules and semantics are included
\begin{techreport}
in the appendix.
\end{techreport}
\begin{cameraready}
in the full version of this paper~\cite{ricciotti21lifting}.
\end{cameraready}

\NRClsb subsumes previous systems including
\NRC~\cite{BNTW95,wong96jcss}, \BQL~\cite{LW97} and
\NRCl~\cite{Cooper09,cheney14sigmod}.  In this paper, we restrict
our attention to queries in which collection types taking part in
$\distinct$, $\promote$ or bag difference contain only flat records. There are various reasons for excluding function types from these operators: for starters, any concrete implementation that used function types in these positions would need to decide the equality of functions; secondly, our rewrite system can ensure that a term whose type does not contain function types has a normal form without lambda abstractions and applications only if any $\distinct$, $\promote$, or bag difference used in that term are applied to first-order collections. We thus want to exclude terms such as:
\[
\biguplus\msetlit{ x~\msetlit{1}~\msetlit{2} | x \gets \promote (\setlit{\lambda yz.y} \cup \setlit{\lambda yz.z})}  
\]
which do not have an SQL representation despite having a flat collection type.

In order to obtain simpler normal forms, in which comprehensions only reference generators with a flat collection type, we also disallow nested collections within $\distinct$, $\promote$, and bag difference. We believe this is without loss of generality because of Libkin and
Wong's results showing that allowing such operations at nested types
does not add expressiveness to \BQL.

We have extended Wong's normalizing rewrite rule system, so as to
simplify queries to a form that is close to SQL, with no intermediate
nested collections. Since our calculus is more liberal than Wong's, allowing queries to be defined by mixing sets and bags and also using bag difference, we have added non-standard rules to take care of unwanted situations. In particular, we use the following constrained eta-expansions for comprehensions:
\begin{align*}
  \comprehension{\distinct(M - N) | \Theta}
& \red \comprehension{\setlit{z} | \Theta, z \gets \distinct(M - N)}
\\
  \mcomprehension{\promote M | \Theta}
& \red \mcomprehension{\msetlit{z} | \Theta, z \gets \promote M}
\\
  \mcomprehension{ M - N | \Theta}
& \red \mcomprehension{\msetlit{z} | \Theta, z \gets M - N}
\end{align*}

The rationale of these rules is that in order to achieve, for comprehensions, a form that can be easily translated to an SQL select query, we need to move all the syntactic forms that are blocking to most normalization rules (i.e. promotion and bag difference) from the head of the comprehension to a generator. In order for this strategy to work out, we also need to know that the type of these subexpressions is flat, as we previously mentioned.

\begin{figure}[tb]
  \[
  \begin{array}{lrcl}
  \textbf{General normal forms} & M
  & ::= &
  X \orelse \tuple{\vect{\ell = M}} \orelse Q \orelse  R

\\
  \textbf{Base type terms}
  & X
  & ::= &
  x.\ell \orelse c(\vect{X}) \orelse \setempty(Q^*) \orelse \bagempty(R^*)
  
\\
  
  \textbf{Set normal forms} & Q
  & ::= &
  \bigcup \vect{C}
  
\\
  
  & C
  & ::= &
  \bigcup\setlit{\setlit{M}~\setwhere~X | \vect{x \gets F}}
  
\\
  
  & F
  & ::= &
 \distinct t \orelse \delta(R^*_1 - R^*_2)
  
\\
  
  \textbf{Bag normal forms} & R
  & ::= & 
  \biguplus\vect{D}
  
\\
  
  & D
  & ::= &
  \biguplus\msetlit{\msetlit{M}~\bagwhere~X | \vect{x \gets G}}
  
\\
  
  & G
  & ::= &
  t \orelse \promote Q^* \orelse R^*_1 - R^*_2
  
  \end{array}
  \]
  
  \caption{Nested relational normal forms.}
  \label{fig:nestrelNF}
\end{figure}

In Figure~\ref{fig:nestrelNF} we show the grammar for the normal forms for terms of \emph{nested relational types}, i.e. types of the following form:
\[
\sigma ::= b \orelse \tuple{\vect{\ell : \sigma}} \orelse \setlit{\sigma} \orelse \msetlit{\sigma}  
\]

For ease of presentation, the grammar actually describes a ``standardized'' version of the normal forms in which:
\begin{itemize}
  \item $\emptyset$ is represented as the trivial union $\bigcup~\vect{C}$ where $\vect{C}$ is the empty sequence; $\emptymset$ has a similar representation using a trivial disjoint union;
  \item comprehensions without a guard are considered to be the same as those with a trivial $\kwtrue$ guard:
  \[
    \comprehension{ \setlit{M} | \Theta} = \comprehension{ \setlit{M}~\kwwhere~\kwtrue \mid \Theta}
  \]
  \item singletons that do not appear as the head of a comprehension are represented as trivial comprehensions:
  \[
    \setlit{M} = \comprehension{ \setlit{M} \mid ~}
  \]
\end{itemize}

Each normal form $M$ can be either a term of base type $X$, a tuple $\tuple{\vect{\ell=M}}$, a set $Q$, or a bag $R$. The normal forms of sets and bags are rather similar, both being defined as unions of comprehensions with a singleton head. The generators for set comprehensions $F$ include deduplicated tables and deduplicated bag differences; the generators for bag comprehensions $G$ must be either tables, promoted set queries, or bag differences.

The non-terminals used as the arguments of emptiness tests, promotion, and bag difference have been marked with a star to emphasize the fact that they must have a flat collection type. The corresponding grammar can be obtained from the grammar for nested normal forms by replacing the rule for $M$ with the following:
\[
M^* ::= \tuple{\vect{\ell = X}}
\]

\begin{figure}[tb]
  \scriptsize{\[
    \begin{array}{rcl@{\quad}rcl}
      \tosql{\emptyset} &=& \kwsel~42~\kwsqlwhere~0=1
      & 
      \tosql{\emptymset} &=& \kwsel~42~\kwsqlwhere~0=1
      \\
      \tosql{x.\ell} &=& x.\ell
      &
      \tosql{c(\vect{X})} &=& \tosql{c}(\vect{\tosql{X}})
      \\
      \tosql{\tuple{\vect{\ell = X}}} & = & \multicolumn{4}{l}{\tosql{X_1} ~\kwas~\ell_1,\ldots,\tosql{X_n}~\kwas~\ell_n}
      \\
      \tosql{\setempty(Q^*)} &=& \kwnot~\kwex~\tosql{Q^*}
      &
      \tosql{\bagempty(R^*)} &=& \kwnot~\kwex~\tosql{R^*}
      \\
        \tosql{Q^*_1 \cup Q^*_2} &=& \tosql{Q^*_1} ~\kwunion~\tosql{Q^*_2}
        &
        \tosql{R^*_1 \uplus R^*_2} &=& \tosql{R^*_1} ~\kwunion~\kwall~\tosql{R^*_2}
        \\
        \tosql{t} &=& \kwsel~*~\kwfrom~t
        &
        \tosql{R^*_1 - R^*_2} &=& \tosql{R^*_1} ~\kwexcept~\kwall~\tosql{R^*_2}
        \\
        \tosql{\distinct t} &=& \kwsel~\kwdist~*~\kwfrom~t
        &
        \tosql{\promote(Q^*)} &=& \tosql{Q^*}
        \\
        \tosql{\distinct (R_1^*-R_2^*)} &=& \multicolumn{4}{l}{\kwsel~\kwdist~*~\kwfrom~(\tosql{R_1^*}~\kwexcept~\kwall~\tosql{R_2^*}s)~r}
        \\
      \tosql{x \leftarrow F} &=& 
        \multicolumn{4}{l}{
        \left\{
        \begin{array}{ll} 
        (\tosql{F})~x & \text{($x$ closed)}\\
        \kwlateral~(\tosql{F})~x & \text{(otherwise)}
        \end{array}\right.
        }
        \\
      \tosql{x \leftarrow G} &=& 
      \multicolumn{4}{l}{\left\{
      \begin{array}{ll} 
      (\tosql{G})~x & \text{($x$ closed)}\\
      \kwlateral~(\tosql{G})~x & \text{(otherwise)}
      \end{array}\right.
      }
        \\
      \multicolumn{6}{c}{
      \tosql{\bigcup\setlit{\setlit{M^*}~\setwhere~X \mid \vect{x \leftarrow
        F}}} = \kwsel~\kwdist~\tosql{M^*}~\kwfrom~\tosql{\vect{x \leftarrow
                F}}~\kwsqlwhere~\tosql{X}
      }
      \\
      \multicolumn{6}{c}{
      \tosql{\biguplus\msetlit{\msetlit{M^*}~ \bagwhere~X \mid \vect{x \leftarrow
        G}}} = \kwsel~\tosql{M^*}~\kwfrom~\vect{\tosql{ x \leftarrow
                G}}~\kwsqlwhere~\tosql{X}
      }
\end{array}
\]}
\caption{Translation to SQL}\label{fig:tosql}
\end{figure}


Normalized queries can be translated to SQL as shown in
Figure~\ref{fig:tosql} as long as they have a flat collection type. The translation uses $\kwsel~\kwdist$ and $\kwunion$ where a set semantics is needed, and $\kwsel$, $\kwunion~\kwall$ and $\kwexcept~\kwall$ in the case of bag semantics.  Note that promotion expressions $\promote Q^*$
are translated simply by translating $Q^*$, because in SQL there is no
type distinction between set and multiset queries: all query results
are multisets, and sets are considered to be multisets having no
duplicates.

The other main complication in this translation is in handling
generators $x \leftarrow F$, $x \leftarrow G$ where $F$ or $G$ may be a non-closed expression
$\promote(Q^*)$, $R^*_1-R^*_2$, or $\distinct(R^*_1 - R^*_2)$ containing references to other
locally-bound variables.  To deal with the resulting lateral variable
references, we add the $\kwlateral$ keyword to such queries.  As
explained earlier, the use of $\kwlateral$ can be problematic and we
will return to this issue in Section~\ref{sec:delat}.

\begin{remark}[Record flattening]
  The above translations handle queries that take flat tables as input
  and produce flat results (collections of flat records
  $\tuple{\vect{\ell:b}}$).  It is straightforward to support queries
  that return nested records (i.e. records containing other records,
  but not collections).  For example, a query
  \linebreak
  $M : \msetlit{\tuple{b_1,\tuple{b_2,b_3}}}$ can be handled by
  defining both directions of the obvious isomorphism
  $N : \msetlit{\tuple{b_1,\tuple{b_2,b_3}}} \cong
  \msetlit{\tuple{b_1,b_2,b_3}}: N^{-1}$,
  normalizing the flat query $N \circ M$, evaluating the
  corresponding SQL, and applying the inverse $N^{-1}$ to the results.
  Such
  \emph{record flattening} is described in detail by Cheney et
  al.~\cite{cheney14tr} and is implemented in Links, so we will use it
  from now on without further discussion.
\end{remark}

\section{A relational calculus of tabular functions}\label{sec:graph}
We now introduce $\NRCg$, an extension of the calculus $\NRClsb$
providing a new type of finite tabular function graphs (in the
remainder of this paper, also called simply ``graphs''; they are similar
to the finite maps and tables of Gibbons et
al.~\cite{gibbons18icfp}). The syntax of $\NRCg$ is defined as
follows:

\begin{tabular}{l@{\hspace{1em}}rcl}
  \textbf{Types} & $\sigma, \tau$ & $::=$ & $\cdots  \orelse
                                            \vect{\sigma} \strictfun
                                            \tau$
  \\
  \textbf{Terms} & $M, N$ & $::=$ & $\cdots \orelse \setgraph{\Theta}{N} \orelse \baggraph{\Theta}{N} \orelse
                                    M \gapp(\vect{N})$
\end{tabular}

Semantically, the type of graphs $\vect{\sigma} \strictfun \tau$ will be interpreted as the set of finite functions from sequences of values of type $\vect{\sigma}$ to values in $\tau$: such functions can return non-trivial values only for a finite subset of their input type. In our settings, we will require the output type of graphs to be a collection type (i.e. $\tau$ shall be either $\setlit{\tau'}$ or $\msetlit{\tau'}$ for some $\tau'$), and we will use $\emptyset$ or $\emptymset$ as the trivial value. The typing rules involving graphs are shown in Figure~\ref{fig:graphtyping}.

\begin{figure}[tb]
  \begin{center}
    \begin{tabular}{c}
      \rVS
        \AxiomC{$(\Gamma, \vect{x_{i-1} : \sigma_{i-1}} \vdash L_i : \setlit{\sigma_i})_{i=1,\ldots,n}$}
        \noLine
        \UnaryInfC{$\Gamma, \vect{x : \sigma} \vdash M : \setlit{\tau}$}
        \UnaryInfC{$\Gamma \vdash \setgraph{\vect{x \gets L}}{M} : \vect{\sigma} \strictfun \setlit{\tau}$}
        \DisplayProof
      \qquad
        \AxiomC{$(\Gamma, \vect{x_{i-1} : \sigma_{i-1}} \vdash L_i : \setlit{\sigma_i})_{i=1,\ldots,n}$}
        \noLine
        \UnaryInfC{$\Gamma, \vect{x : \sigma} \vdash M : \msetlit{\tau}$}
        \UnaryInfC{$\Gamma \vdash \baggraph{\vect{x \gets L}}{M} : \vect{\sigma} \strictfun \msetlit{\tau}$}
        \DisplayProof
      \\
      \rVS
        \AxiomC{$\Gamma \vdash M : \vect{\sigma} \strictfun \tau$}
        \AxiomC{$(\Gamma \vdash N_i : \sigma_i)_i$}
        \BinaryInfC{$\Gamma \vdash M \gapp(\vect{N}) : \tau$}
        \DisplayProof
      \qquad
        \AxiomC{$\Gamma \vdash M : \vect{\sigma} \strictfun \msetlit{\tau}$}
        \noLine
        \UnaryInfC{$\Gamma \vdash N : \vect{\sigma} \strictfun \msetlit{\tau}$}
        \UnaryInfC{$\Gamma \vdash M - N : \vect{\sigma} \strictfun \msetlit{\tau}$}
        \DisplayProof
      \\
      \rVS
        \AxiomC{$\Gamma \vdash M : \vect{\sigma} \strictfun \setlit{\tau}$}
        \noLine
        \UnaryInfC{$\Gamma \vdash N : \vect{\sigma} \strictfun \setlit{\tau}$}
        \UnaryInfC{$\Gamma \vdash M \cup N : \vect{\sigma} \strictfun \setlit{\tau}$}
        \DisplayProof
      \qquad
        \AxiomC{$\Gamma \vdash M : \vect{\sigma} \strictfun \msetlit{\tau}$}
        \noLine
        \UnaryInfC{$\Gamma \vdash N : \vect{\sigma} \strictfun \msetlit{\tau}$}
        \UnaryInfC{$\Gamma \vdash M \uplus N : \vect{\sigma} \strictfun \msetlit{\tau}$}
        \DisplayProof
      \\
      \rVS
        \AxiomC{$\Gamma \vdash M : \vect{\sigma} \strictfun \msetlit{\tau}$}
        \UnaryInfC{$\Gamma \vdash \distinct M : \vect{\sigma} \strictfun \setlit{\tau}$}
        \DisplayProof
      \qquad
        \AxiomC{$\Gamma \vdash M : \vect{\sigma} \strictfun \setlit{\tau}$}
        \UnaryInfC{$\Gamma \vdash \promote M : \vect{\sigma} \strictfun \msetlit{\tau}$}
        \DisplayProof


    \end{tabular}
  \end{center}
  \caption{$\NRCg$ additional typing rules.}\label{fig:graphtyping}
\end{figure}

Graphs are created using the \emph{graph abstraction} operations $\setgraph{\Theta}{N}$ and $\baggraph{\Theta}{N}$, where $\Theta$ is a sequence of generators in the form $\vect{x \gets M}$; the dual operation of \emph{graph application} is denoted by $M \gapp(\vect{N})$. An expression of the form $\setgraph{\vect{x \gets M}}{N}$ is used to construct a (finite) tabular function mapping 
each sequence of values $R_1,\ldots,R_n$ in the sets $M_1,\ldots,M_n$ to the set $N\subst{\vect{x}}{\vect{R}}$.
If each $M_i$ has type $\setlit{\sigma_i}$ and $N$ has type $\setlit{\tau}$, then the graph
has type $\vect{\sigma} \strictfun \setlit{\tau}$. Similarly, if $N$ has type $\msetlit{\tau}$, $\baggraph{\vect{x \gets M}}{N}$ has type $\vect{\sigma} \strictfun \msetlit{\tau}$. The terms $M_1,\ldots,M_n$ constitute the (finite) \emph{domain} of this graph. When the kind of graph application (set-based or bag-based) is clear from the context or unimportant, we will allow ourselves to write $\graph{-}{-}$ instead of $\setgraph{-}{-}$ or $\baggraph{-}{-}$.

A graph $G$ of type $\vect{\sigma} \strictfun \tau$ can be applied to a sequence of terms $N_1,\ldots, N_n$ of type $\sigma_1, \ldots, \sigma_n$ to obtain a term of type $\tau$. If $G = \graph{\vect{x \gets L}}{M}$, then we will want the semantics of $\graph{\vect{x \gets L}}{M} \gapp(\vect{N})$ to be the same as that of $M\subst{\vect{x}}{\vect{N}}$, provided that each of the $N_i$ is in the corresponding element of the domain of the graph. The typing rule does \emph{not}
enforce this requirement and if any of the $N_i$ is not an element of $L_i$, the graph application will evaluate to an empty set or bag (depending on $\tau$).

Graphs can also be merged by union, using $\cup$ or $\uplus$ depending on their output collection kind. Furthermore, graphs that return bags can be subtracted from one another using bag difference; the deduplication and promotion operations also extend to graphs in the obvious way.

\begin{lemma}
In \NRCg, $\Gamma \vdash M : \sigma$ and $\Gamma \vdash M : \tau$, then $\sigma = \tau$.
\end{lemma}

Whenever $M$ is well typed and its typing environment is made clear by the context, we will allow ourselves to write $ty(M)$ for the type of $M$. Furthermore, given a sequence of generators $\Theta = x_1 \gets L_1, \ldots x_n \gets L_n$, such that for $i = 1, \ldots, n$ we have $x_1 : \sigma_1, \ldots, x_{i-1} : \sigma_{i-1} \vdash L_i : \sigma_i$, we will write $ty(\Theta)$ to denote the associated typing context:
\[
  ty(\Theta) := x_1 : \sigma_1, \ldots, x_n : \sigma_n
\]

\subsection{Semantics and translation to \NRClsb}
The semantics of \NRClsb is extended to \NRCg as follows: 

\[
  \begin{array}{l}
    \sem{\setgraph{\vect{x \gets L}}{M}}\rho(\vect{u},v)
    \\
    \quad = \left(\bigwedge_i \sem{L_i}\rho[x_1 \mapsto u_1,\ldots,x_{i-1} \mapsto u_{i-1}] u_i\right) \land \sem{M}\rho[\vect{x \mapsto u}] v
    \\
    \sem{\baggraph{\vect{x \gets L}}{M}}\rho(\vect{u},v)
    \\
    \quad = \left(\bigwedge_i \sem{L_i}\rho[x_1 \mapsto u_1,\ldots,x_{i-1} \mapsto u_{i-1}] u_i\right) \times \sem{M}\rho[\vect{x \mapsto u}] v
    \\
    \sem{M \gapp(\vect{N})}\rho v
    = \sem{M}\rho~(\vect{\sem{N}\rho}, v)
  \end{array}
\]

In this definition, graph abstractions are interpreted as collections of pairs of values $(\vect{u}, v)$ where the $\vect{u}$ represent the input and $v$ the corresponding output of the graph; consequently, the semantics of a graph $\setgraph{\vect{x \gets L}}{M}$ states that the multiplicity of $(\vect{u}, v)$ is equal to the multiplicity of $v$ in the semantics of $M$ (where each $x_i$ is mapped to $u_i$) if each $u_i$ is in the semantics of $L_i$, and zero otherwise. The semantics of bag graph abstractions is similar, with $\times$ substituted for $\land$ to allow multiplicities greater than one in the graph output.

For graph applications $M \gapp(\vect{N})$, the multiplicity of $v$ is obtained as the multiplicity of $(\vect{\sem{N}\rho},v)$ in the semantics of $M$. The semantics of set and bag union, bag difference, bag deduplication, and set promotion, as defined in 
\linebreak
\NRClsb, are extended to graphs and remain otherwise unchanged in \NRCg.

In fact (as noted for example by Gibbons et al.~\cite{gibbons18icfp}),
the graph constructs of \NRCg are just a notational convenience: we
can translate \NRCg back to \NRClsb by translating types
$\vect{\sigma} \strictfun \setlit{\tau}$ and
$\vect{\sigma} \strictfun \msetlit{\tau}$ to
$\setlit{\tuple{\vect{\sigma},\tau}}$ and
$\msetlit{\tuple{\vect{\sigma},\tau}}$ respectively, and the term
constructs are rewritten as follows:
\begin{eqnarray*}
  \setgraph{\vect{x\leftarrow L}}{M} & \leadsto &
                                                  \bigcup\setlit{\setlit{\tuple{\vect{x},y}
                                                  } \mid \vect{x
                                                  \leftarrow L}, y
                                                  \leftarrow M}\\
  \baggraph{\vect{x\leftarrow L}}{M} & \leadsto &
                                                  \biguplus\msetlit{\msetlit{\tuple{\vect{x},y}
                                                  } \mid \vect{x
                                                  \leftarrow \promote(L)}, y
                                                  \leftarrow M}\\
M \gapp \tuple{\vect{N}} &\leadsto & \bigcup\setlit{\setlit{y}
                                           ~\setwhere \vect{x} = \vect{N}
                                           \mid \tuple{\vect{x},y}
                                           \leftarrow M} \quad (M : \vect{\sigma} \strictfun \setlit{\tau})\\
M \gapp \tuple{\vect{N}} &\leadsto & \biguplus\msetlit{\msetlit{y}
                                           ~\bagwhere \vect{x} = \vect{N}
                                           \mid \tuple{\vect{x},y}
                                           \leftarrow M} \quad (M : \vect{\sigma} \strictfun \msetlit{\tau})
\end{eqnarray*}

\section{Delateralization}\label{sec:delat}
As explained at the end of section~\ref{sec:background}, if a
subexpression of the form $\promote(N)$ or $N_1-N_2$ contains free
variables introduced by other generators in the query (i.e. not
globally-scoped table variables), such queries cannot be translated
directly to SQL, unless the SQL:1999 $\kwlateral$ keyword is used. 

More precisely, we can give the following definition of lateral variable occurrence.
\begin{definition}
  Given a query containing a comprehension $\comprehension{M \mid \Theta, x \gets N, \Theta'}$ or $\mcomprehension{M \mid \Theta, x \gets N, \Theta'}$ as a subterm, we say that $x$ \emph{occurs laterally} in $\Theta'$ if, and only if, there is a binding $y \gets N'$ in $\Theta'$ such that $x \in \FV(N')$.
\end{definition}

Since $\kwlateral$ is not implemented on all databases, and is sometimes
implemented inefficiently, we would still like to avoid it. In this section
we show how lateral occurrences can be eliminated even in the presence of bag promotion and bag difference, by means of a process we call \emph{delateralization}.

Using the \NRCg constructs, we can delateralize simple cases of
deduplication or multiset difference as follows:
{\small\[\begin{array}{l}
\biguplus\msetlit{M \mid x \leftarrow N, y \leftarrow \promote(P)} \leadsto
\biguplus\msetlit{M \mid x \leftarrow N, y \leftarrow 
  \promote(\graph{x \leftarrow \distinct N}{P}) \gapp x}
  \\
\biguplus\msetlit{M \mid x \leftarrow N, y \leftarrow P_1-P_2} \leadsto ~
\\
  \qquad
  \biguplus\msetlit{M \mid x \leftarrow N, y \leftarrow
(\graph{x \leftarrow \distinct{N}}{P_1} - \graph{x \leftarrow
                                                                          \distinct{N}}{P_2})\gapp
                                                                          x}
\\
\bigcup\setlit{M \mid x \leftarrow N, y \leftarrow \distinct(P_1-P_2)} \leadsto
\\
  \qquad
\bigcup\setlit{M \mid x \leftarrow N, y \leftarrow
\distinct(\graph{x \leftarrow N}{P_1} - \graph{x \leftarrow
                                                                          N}{P_2})\gapp x}
\end{array}\]}
It is necessary to deduplicate $N$ in the first two rules to ensure that the results correctly represent
finite maps from the distinct elements of $N$ to multisets of
corresponding elements
of $P$.  (In any case, $N$ needs to be deduplicated in order to be
used as a set in $\graph{x \leftarrow \distinct N}{\_}$).

Given a query expression in normal form, the above rules together with
standard equivalences (such as commutativity of independent
generators) can be used to delateralize it: that is, remove all
occurrences of free variables in subexpressions of the form
$\promote(N)$, $M_1-M_2$, or $\distinct(M_1-M_2)$.
\begin{theorem}
  If $M$ is a flat query in normal form, then there exists $M'$
  equivalent to $M$ with no lateral variable occurrences.
\end{theorem}
The proof of correctness of the basic delateralization rules and the
above correctness theorem are
\begin{techreport}
in the appendix.
\end{techreport}
\begin{cameraready}
in the full version of this paper~\cite{ricciotti21lifting}.
\end{cameraready}  

To illustrate some subtleties of the translation, here is a
trickier example:
\[\biguplus\msetlit{M \mid x \leftarrow N, y \leftarrow
  Q - \promote(P)}\]
where $Q,P$ both depend on $x$.  We proceed from the outside in, first
delateralizing the difference:
\[\biguplus\msetlit{M \mid x \leftarrow N, y \leftarrow
  (\graph{x \leftarrow \distinct(N)}{Q} - \graph{x
    \leftarrow\distinct(N)}{\promote(P)})\gapp x}\]
Note that this still contains a lateral subquery, namely
$\promote(P)$ depends on $x$.  After translating back to
$\NRClsb$, and delateralizing $\promote(P)$, the query normalizes to:
{\small\[\begin{array}{rcl}
Q_1 &=& \bigcup\setlit{(x,z) \mid x \in \distinct(N),z
  \leftarrow P}\\
 Q_2 &=& (\biguplus\msetlit{(x,z) \mid x \in \promote\distinct(N),z \leftarrow
  Q}) - (\biguplus\msetlit{(x,z) \mid x \in \promote\distinct(N),(x',z)
  \leftarrow \promote(Q_1),x=x'})\\
&& \biguplus\msetlit{M \mid x \leftarrow N, (x',y) \leftarrow Q_2
  , x=x'}
\end{array}\]}

\section{Query lifting and shredding}\label{sec:shredding}
  In the previous sections, we have discussed how to translate queries with flat collection input and output to SQL. The shredding technique, introduced in~\cite{cheney14sigmod}, can be used to convert queries with \emph{nested} output (but flat input) to multiple flat queries that can be independently evaluated on an SQL database, then stitched together to obtain the required nested result. This section provides an improved version of shredding, extended to a more liberal setting mixing sets and bags and allowing bag difference operations, and described using the graph operations we have introduced, allowing an easier understanding of the shredding process.

  We introduce, in Figure~\ref{fig:jshredding}, a \emph{shredding judgment} to denote the process by which, given a normalized \NRClsb query, each of its subqueries having a nested collection type is lifted (in a manner analogous to lambda-lifting~\cite{johnsson85fpca}) to an independent graph query: more specifically, shredding will produce a \emph{shredding environment} (denoted by $\Phi, \Psi, \ldots$), which is a finite map associating special \emph{graph variables} $\varphi, \psi$ to $\NRCg$ terms:
  \[
    \Phi, \Psi, \ldots ::= [\vect{\varphi \mapsto M}]
  \]
  The shredding judgment has the following form:
  \[
    \Phi; \Theta \vdash M \shreds \flat{M} \orelse \Psi
  \]
  where the $\shreds$ symbol separates the input (to the left) from the output (to the right).
  The normalized \NRClsb term $M$ is the query that is being considered for shredding; $M$ may contain free variables declared in $\Theta$, which must be a sequence of $\NRClsb$ set comprehension bindings. $\Theta$ is initially empty, but during shredding it is extended with parts of the input that have already been processed. Similarly, the input shredding environment $\Phi$ is initially empty, but will grow during shredding to collect shredded queries that have already been generated. It is crucial, for our algorithm to work, that $M$ be in the form previously described in Figure~\ref{fig:nestrelNF}, as this allows us to make assumptions on its shape: in describing the judgment rules, we will use the same metavariables as are used in that grammar.

  The output of shredding consists of a shredded term $\flat{M}$ and an output shredding environment $\Psi$. $\Psi$ extends $\Phi$ with the new queries obtained by shredding $M$; $\flat{M}$ is an output \NRCg query obtained from $M$ by lifting its collection typed subqueries to independent queries defined in $\Psi$.

  \begin{figure}[t]
    \scriptsize{
      \begin{center}
      \begin{tabular}{c}
        \rVS
          \AxiomC{$X$ is a base term}
          \UnaryInfC{$\Phi; \Theta \vdash X \shreds X \orelse \Phi$}
          \DisplayProof
        \qquad
            \AxiomC{$(\Phi_{i-1}; \Theta \vdash M_i \shreds \flat{M}_i \orelse \Phi_i)_{i=1,\ldots,n}$}
            \UnaryInfC{$\Phi_0; \Theta \vdash \tuple{\vect{\ell = M}} \shreds \tuple{\vect{\ell = \flat{M}}} \orelse \Phi_n$}
            \DisplayProof
        \\
        \\
        \rVS
          \AxiomC{$\varphi \notin \dom(\Phi_n)$}
          \noLine
          \UnaryInfC{$(\Phi_{i-1}; \Theta \vdash C_i \shreds \psi_i \gapp \dom(\Theta) \orelse \Phi_i)_{i=1,\ldots,n}$}
          \UnaryInfC{$
            \begin{array}{rl}
              \Phi_0; \Theta \vdash & \bigcup \vect{C} \shreds \varphi \gapp \dom(\Theta)
              \\
              \orelse & (\Phi_n \setminus \vect{\psi})[\varphi \mapsto \bigcup\vect{\Phi_n(\psi)}]
            \end{array}$}
          \DisplayProof
        \;
          \AxiomC{$\varphi \notin \dom(\Phi_n)$}
          \noLine
          \UnaryInfC{$(\Phi_{i-1};\Theta \vdash D_i \shreds \psi_i \gapp \dom(\Theta) \orelse \Phi_i)_{i=1,\ldots,n}$}
          \UnaryInfC{$
            \begin{array}{rl}
              \Phi_0;\Theta \vdash & \biguplus \vect{D} \shreds \varphi  \gapp \dom(\Theta) 
              \\
              \orelse & (\Phi_n \setminus \vect{\psi})[\varphi \mapsto \biguplus\vect{\Phi_n(\psi)}]
            \end{array}$}
          \DisplayProof
        \\
        \rVS
          \AxiomC{$\varphi \notin \dom(\Psi)$}
          \noLine
          \UnaryInfC{$\Phi; \Theta, \vect{x \gets F} \vdash M \shreds \flat{M} \orelse \Psi$}
          \UnaryInfC{$
            \begin{array}{rl}
              \Phi; \Theta \vdash & \bigcup\setlit{\setlit{M}~\plwhere~X| \vect{x \gets F}} \shreds \varphi \gapp \dom(\Theta)
              \\
              \orelse & \Psi[\varphi \mapsto \graph{\Theta}{\bigcup\setlit{\setlit{\flat{M}}~\plwhere~X| \vect{x \gets F}}}]
            \end{array}$
          }
          \DisplayProof
        \\
        \rVS
          \AxiomC{$\varphi \notin \dom(\Psi)$}
          \noLine
          \UnaryInfC{$\Phi; \Theta, \vect{x \gets G^\distinct} \vdash M \shreds \flat{M} \orelse \Psi$}
          \UnaryInfC{$
            \begin{array}{rl}
              \Phi_0; \Theta \vdash & \biguplus\msetlit{\msetlit{M}~\plwhere~X | \vect{x \gets G}} \shreds \varphi \gapp \dom(\Theta)
              \\
              \orelse & \Psi[\varphi \mapsto \graph{\Theta}{\biguplus\msetlit{\msetlit{\flat{M}}~\plwhere~X | \vect{x \gets G}}}]
            \end{array}$
          }
          \DisplayProof
      \end{tabular}
      \[
      \begin{array}{c@{\qquad}c}
      G^\distinct \triangleq \left\{
        \begin{array}{l@{\quad}l}
          Q^* & \text{if $G = \promote Q^*$}
          \\
          \distinct G & \text{otherwise}
        \end{array}
        \right.
      &
      \Phi \setminus \vect{\psi} \triangleq [ (\varphi \mapsto N) \in \Phi \mid \varphi \notin \vect{\psi} ]  
      \end{array}
      \]
    \end{center}}
    \caption{Shredding rules.}\label{fig:jshredding}
  \end{figure}

  The rules for the shredding judgment operate as follows: the first rule expresses the fact that a normalized base term $X$ does not contain
  subexpressions with nested collection type, therefore it can be
  shredded to itself, leaving the shredding environment $\Phi$ unchanged; in the case of tuples, we
  perform shredding pointwise on each field, connecting the input and output shredding environments in a pipeline, and finally combining together the shredded subterms in the obvious way.
  
  The shredding of collection terms (i.e. unions and comprehensions) is performed by means of \emph{query lifting}: we turn the collection into a globally defined (graph) query, which will be associated to a fresh name $\varphi$ and instantiated to the local comprehension context by graph application. This operation is reminiscent of the lambda lifting and closure conversion techniques used in the implementation of functional languages to convert local function definitions into global ones. Thus, when shredding a collection, besides processing its subterms recursively, we will need to extend the output shredding environment with a definition for the new global graph $\varphi$. In the interesting case of comprehensions, $\varphi$ is defined by graph-abstracting over the comprehension context $\Theta$; notice that, since we are only shredding normalized terms, we know that they have a certain shape and, in particular, the judgment for bag comprehensions must ensure that generators $\vect{G}$ be converted into sets.

  The shredding of set and bag unions is performed by recursion on the subterms, using the same plumbing technique we employed for tuples; additionally, we optimize the output shredding environment by removing the graph queries $\vect{\psi}$ resulting from recursion, since they are absorbed into the new graph $\varphi$.

  Notice that since the comprehension generators of our normalized queries must have a flat collection type, they do not need to be processed recursively. Furthermore, since our normal forms ensure that promotion and bag difference terms can only appear as comprehension generators, we do not need to provide rules for these cases.

  \begin{figure}[tb]
    \begin{center}
        \AxiomC{\phantom{$A$}}
        \UnaryInfC{$\vdash \cdot : \cdot$}
        \DisplayProof
      \hspace{0.5cm}
        \AxiomC{$\vdash \Phi : \Gamma$}
        \AxiomC{$\Gamma \vdash M : \vect{\sigma} \strictfun \tau$}
        \AxiomC{$\varphi \notin \dom(\Gamma)$}
        \TrinaryInfC{$\vdash \Phi[\varphi \mapsto M] : (\Gamma, \varphi:\vect{\sigma} \strictfun \tau)$}
        \DisplayProof
    \end{center}
    \caption{Typing rules for shredding environments.}\label{fig:typPhi}
  \end{figure}

  The shredding environments used by the shredding judgment must be well typed, in the sense described by the rules of Figure~\ref{fig:typPhi}: the judgment $\vdash \Phi : \Gamma$ means that the graph variables of $\Phi$ are mapped to terms whose type is described by $\Gamma$. Whenever we add a mapping $[\varphi \mapsto M]$ to $\Phi$, we must make sure that $M$ is well typed (of graph type) in the typing environment $\Gamma$ associated to $\Phi$.

  If $\vdash \Phi : \Gamma$, we will write $ty(\Phi)$ to refer to the typing environment $\Gamma$ associated to $\Phi$. The following result states that shredding preserves well-typedness:

  \begin{theorem}
    Let $\Theta$ be well-typed and $ty(\Theta) \vdash M : \sigma$. If $\Theta \vdash M \shreds \flat{M} \orelse \Phi$, then:
    \begin{itemize}
    \item $\Phi$ is well-typed
    \item $ty(\Phi), ty(\Theta) \vdash \flat{M} : \sigma$
    \end{itemize}
  \end{theorem}

  We now intend to prove the correctness of shredding: first, we state a lemma which we can use to simplify certain expressions involving the semantics of graph application:

  \begin{definition}
    Let $\Theta$ be a closed, well-typed sequence of generators. A substitution $\rho$ is a model of $\Theta$ (notation: $\rho \vDash \Theta$) if, and only if, for all $x \in \dom(\Theta)$, we have $\sem{\Theta(x))}\rho(x) > 0$.
  \end{definition}
  
  \begin{lemma}\label{lem:grapheqn}
    \begin{enumerate}
      \item $\sem{(\bigcup \vect{G}) \gapp(\vect{N})}\rho = \bigvee_i \sem{G_i \gapp(\vect{N})}\rho$
      \item If $\rho \vDash \Theta$, then for all $M$ we have $\sem{\graph{\Theta}{M} \gapp(\dom(\Theta))}\rho = \sem{M}\rho$.
    \end{enumerate}
  \end{lemma}

 To state the correctness of shredding, we need the following notion of shredding environment substitution.
 \begin{definition}
    For every well-typed shredding environment $\Phi$, the substitution of $\Phi$ into an \NRCg term $M$ (notation: $M\Phi$) is defined as the operation replacing within $M$ every free variable $\varphi \in \dom(\Phi)$ with $(\Phi(\varphi))\Phi$ (i.e.: the value assigned by $\Phi$ to $\varphi$, after recursively substituting $\Phi$).
  \end{definition}
  We can easily show that the above definition is well posed for well-typed $\Phi$.

  We now show that shredding preserves the semantics of the input term, in the sense that the term obtained by substituting the output shredding environment into the output term is equivalent to the input.

  \begin{theorem}[Correctness of shredding]\label{thm:shredsound}
    Let $\Theta$ be well-typed and $ty(\Theta) \vdash M : \sigma$. If $\Phi; \Theta \vdash M \shreds \flat{M} \orelse \Psi $, 
    then, for all $\rho \vDash \Theta$, we have $\sem{M}\rho = \sem{\flat{M}\Psi}\rho$.
  \end{theorem}
  \begin{proof}
    By induction on the shredding judgment. We comment two
    representative cases:
    \begin{itemize}
      \item in the set comprehension case, we want to prove 
      \[
        \begin{array}{l}
        \sem{\bigcup\setlit{\setlit{M}~\kwwhere~X| \vect{x \gets F]}}}\rho~v= 
        \\
        \quad \sem{(\varphi \gapp(\dom(\Theta)))\Psi[\varphi \mapsto \bigcup\setlit{\graph{\Theta}{\bigcup\setlit{\setlit{\flat{M}}~\kwwhere~X|\vect{x \gets F}}}}]}\rho~v
        \end{array}
      \]
      where $\rho \vDash \Theta$. We rewrite the lhs as follows:
      \[
      \begin{array}{l}
        \sem{\bigcup\setlit{\setlit{M}~\kwwhere~X | \vect{x \gets F]}}}\rho~v
        \\
        = \bigvee_{\vect{u}} (\sem{M}\rho_n = v) \land (\sem{X} \rho_n) \land (\sem{F_i}\rho_{i-1}~u_i))_{i=1,\ldots,n}
      \end{array}
      \]
      where $\rho_i = \rho[x_1 \mapsto u_1, \ldots, x_i \mapsto u_i] \vDash \Theta, x_1 \gets F_1, \ldots, x_i \gets F_i$ for all $i = 1, \ldots, n$, and $u_i$ s.t. $\sem{F_i}\rho_{i-1} u_i$. 
      By the definition of substitution and by Lemma~\ref{lem:grapheqn}, we rewrite the rhs:
      \[
        \begin{array}{l}
          \sem{(\varphi \gapp(\dom(\Theta)))\Psi[\varphi \mapsto \graph{\Theta}{\bigcup\setlit{\setlit{\flat{M}}~\kwwhere~X|\vect{x \gets F}}}]}\rho~v
          \\
          = \sem{(\graph{\Theta}{\bigcup\setlit{\setlit{\flat{M}\Psi}~\kwwhere~X|\vect{x \gets F}}}) \gapp(\dom(\Theta))}\rho~v
          \\
          = \sem{\bigcup\setlit{\setlit{\flat{M}\Psi}~\kwwhere~X|\vect{x \gets F}}}\rho~v
          \\
          = \bigvee_{\vect{u}} (\sem{\flat{M}\Psi}\rho_n = v) \land (\sem{F_i}\rho_{i-1}~u_i))_{i=1,\ldots,n} \land (\sem{X} \rho')
          \end{array}
      \]
      We can prove that for all $\vect{u}$ such that $\rho_n \nvDash \Theta, \vect{x \gets F}$, $(\sem{F_i}\rho_{i-1}~u_i)_{i=1,\ldots,n} = 0$. Therefore, we only need to consider those $\vect{u}$ such that $\rho_n \vDash \Theta, \vect{x \gets F}$.

      Then, to prove the thesis, we only need to show:
      \[
        \sem{M}\rho_n = \sem{\flat{M}\Phi}\rho_n
      \]
      which follows by induction hypothesis, for $\rho_n \vDash \Theta, \vect{x \gets F}$.
      \item in the set union case, we want to prove
      \[
        \sem{\bigcup\vect{C}}\rho~v
        = 
        \sem{(\varphi \gapp (\dom(\Theta)))(\Psi \setminus \vect{\psi})[\varphi \mapsto \bigcup \vect{\Psi(\psi))}]}\rho~v
      \]
      where $\rho \vDash \Theta$. We rewrite the lhs as follows:
      \[
        \sem{\bigcup\vect{C}}\rho~v
        = \bigvee_i \sem{C_i}\rho~v
      \] 
      By the definition of substitution and by Lemma~\ref{lem:grapheqn}, we rewrite the rhs:
      \[
        \begin{array}{l}
          \sem{(\varphi \gapp (\dom(\Theta)))(\Psi \setminus \vect{\psi})[\varphi \mapsto \bigcup \vect{\Psi(\psi))}]}\rho~v 
          \\
          = \sem{(\bigcup\vect{(\Psi(\psi))\Psi})\gapp(\dom(\Theta))}\rho~v
          \\
          = \bigvee_i \sem{(\Psi(\psi_i))\Psi\gapp(\dom(\Theta))}\rho~v
        \end{array}
      \]
      By induction hypothesis and unfolding of definitions, we know for all $i$:
      \[
      \sem{C_i}\rho = \sem{(\psi_i \gapp(\vect{\dom(\Theta)}))\Psi}\rho
      = \sem{(\Psi(\psi_i))\Psi \gapp(\vect{\dom(\Theta)})}\rho
      \]
      which proves the thesis.\qed
    \end{itemize}
  \end{proof}

  \subsection{Reflecting shredded queries into \NRClsb}
  The output of the shredding judgment is a stratified version of the input term, where each element of the output shredding environment provides a layer of collection nesting; furthermore, the output is ordered so that each element of the shredding environment only references graph variables defined to its left, which is convenient for evaluation. Our goal is to evaluate each shredded item as an independent query: however, these items are not immediately convertible to flat queries, partly because their type is still nested, and also due to the presence of graph operations introduced during shredding. We thus need to provide a translation operation capable of converting the output of shredding into independent flat terms of \NRClsb. This translation uses two main ingredients:
  \begin{itemize}
    \item an $\kwindex$ function to convert graph variable references
      to a flat type $\bbI$ of indices, such that $\phi,\vect{x}$ are
      recoverable from $\kwindex(\phi,\vect{x})$;
    \item a technique to express graphs as standard \NRClsb relations.
  \end{itemize}

  The resulting translation, denoted by $\defun{\cdot}$, is shown in in Figure~\ref{fig:defun}. Let us remark that the translation need be defined only for term forms that can be produced as the output of shredding: this allows us, for instance, not to consider terms such as $\promote M$ or $M - N$, which can only appear as part of flat generators of comprehensions or graphs.
  
  We discuss briefly the interesting cases of the definition of the flattening translation. Base expressions $X$ are expressible in $\NRClsb$, therefore they can be mapped to themselves (this is also true for $\kwempty(M)$, since normalization ensures that the type of $M$ be a flat collection). Graph applications $\varphi \gapp(\vect{x})$, as we said, are translated with the help of an $\kwindex$ abstract operation: this is where the primary purpose of the translation is accomplished, by flattening a collection type to the flat type $\bbI$, making it possible for a shredded query to be converted to SQL; although we do not specify the concrete implementation of $\kwindex$, it is worth noting that it must store the arguments of the graph application along with the (quoted) \emph{name} of the graph variable $\varphi$. Tuples, unions, and comprehensions only require a recursive translation of their subterms: however the generators of comprehensions must have a flat collection type, so no recursion is needed there. Finally, we translate graphs as collections of the pairs obtained by associating elements of the domain of the graph to the corresponding output; it is simple to come up with a comprehension term building such a collection: set-valued graphs are translated using set comprehension, while bag-valued ones use bag comprehension (this also means that in the latter case the generators for the domain of the graph, which are set-typed, must be wrapped in a $\promote$).

  \begin{figure}[tb]
    \scriptsize{
    \begin{align*}
        \defun{X} 
        & = X
      &
        \defun{\tuple{\vect{\ell = M}}}
        & = \tuple{\vect{\ell = \defun{M}}}
        \\
        \defun{\bigcup \vect{C}}
        & = \bigcup \vect{\defun{C}}
      &
        \defun{\biguplus \vect{D}}
        & = \biguplus \vect{\defun{D}}
      \\
        \defun{\varphi \gapp(\vect{x})}
        & = \kwindex(\varphi,\vect{x})
    \end{align*}
    \begin{align*}
        \defun{\bigcup\setlit{\setlit{M}~\plwhere~X| \vect{x \gets F}}}
        & = \bigcup\setlit{\setlit{\defun{M}}~\plwhere~X | \vect{x \gets F}}
      \\
        \defun{\biguplus\msetlit{\msetlit{M}~\plwhere~X | \vect{x \gets G}}}
        & = \biguplus\msetlit{\msetlit{\defun{M}}~\plwhere~X | \vect{x \gets G}}
      \\
        \defun{\setgraph{\vect{x \gets F}}{M}}
        & = \bigcup\setlit{\tuple{x,y} | \vect{x \gets F}, y \gets \defun{M}}
      \\
        \defun{\baggraph{\vect{x \gets F}}{M}}
        & = \biguplus\msetlit{\tuple{x,y} | \vect{x \gets \promote F}, y \gets \defun{M}}
    \end{align*}
    }
    \caption{Flattening embedding of shredded queries into \NRClsb.}\label{fig:defun}
  \end{figure}

  We can prove that the flattening embedding produces flat-typed terms, as expected.
  \begin{definition}
    A well-typed set comprehension generator $\Theta$ is flat-typed if, and only if, for all $x \in \dom(\Theta)$, there exists a flat type $\sigma$ such that $ty(\Theta(x)) = \setlit{\sigma}$.

    A well-typed shredding environment $\Phi$ is flat-typed if, and only if, for all $\varphi \in \dom(\Phi)$, we have that $ty(\defun{\Phi(\varphi)})$ is a flat collection type.
  \end{definition}
  \begin{lemma}
    Suppose $\Phi;\Theta \vdash M \shreds \flat{M} \orelse \Psi$, where $\Phi$ and $\Theta$ are flat-typed. Then, $\flat{M}$ and $\Psi$ are also flat-typed.
  \end{lemma} 

  It is important to note that the composition of shredding and $\defun{\cdot}$ does not produce normalized \NRClsb terms: when we shred a comprehension, we add to the output shredding environment a graph returning a comprehension, and when we translate this to \NRClsb we get two nested comprehensions:
  \[
    \defun{\graph{x \gets \distinct t}{\biguplus\msetlit{\msetlit{\flat{M}}| y \gets \promote Q^*}}} = \biguplus\msetlit{\tuple{x,z} | x \gets \promote\distinct t, z \gets \biguplus\msetlit{\msetlit{\defun{\flat{M}}} | y \gets \promote Q^*}}
  \]
  In fact, not only is this term not in normal form, but it may even
  contain, within $Q^*$, a lateral reference to $x$; thus, after a
  flattening translation, we will always require the resulting queries
  to be renormalized and, if needed, delateralized.

  Let $\norm$ denote \NRClsb normalization, and $\cS$ denote the evaluation of relational normal forms: we define the shredded value set $\Xi$ corresponding to a shredding environment $\Phi$ as follows:
  \[
  \Xi \triangleq \setlit{\varphi \mapsto \cS(\norm(\defun{M})) | [\varphi \mapsto M] \in \Phi}  
  \]
  The evaluation $\cS$ is ordinarily performed by a DBMS after
  converting the 
  \linebreak
  \NRClsb query to SQL, as described
  in~Section~\ref{sec:delat}. The result of this evaluation is reflected in a programming language such as Links as a list of records.

  \subsection{The stitching function}

  Given a \NRClsb term with nested collections, we have first shredded it, obtaining a shredded \NRCg term $\flat{M}$ and a shredding environment $\Phi$ containing \NRCg graphs; then we have used a flattening embedding to reflect both $\flat{M}$ and $\Phi$ back into the flat fragment of \NRClsb; next we used normalization and DBMS evaluation to convert the shredding environment into a shredded value set $\Xi$.
  As the last step to evaluate $M : \tau$, we need to combine
  $\defun{\flat{M}}$ and $\Xi$ together to reconstruct the correct
  nested value $\stitch{\Xi}{\defun{\flat{M}}:\tau}$ by \emph{stitching} together partial flat values.

  \begin{figure}[tb]
    \[
    \scriptstyle{
  \begin{array}{rcl}
      \stitch{\Xi}{X:b}
    & \triangleq &
    X \qquad \text{(if $X$ is not an index)}
    \\
      \stitch{\Xi}{\tuple{\vect{\ell = \flat{N}}} : \tuple{\vect{\ell:\tau}}} 
    & \triangleq &
      \tuple{\vect{\ell = \stitch{\Xi}{\flat{N}:\tau}}}
    \\
      \stitch{\Xi}{\tuple{\vect{\ell = \flat{N}}}.\ell_i : \tau}
    & \triangleq &
      \stitch{\Xi}{N_i : \tau}
    \\
      \stitch{\Xi}{\kwindex(\varphi,\vect{V}):\setlit{\tau}}
    & \triangleq &
      \bigcup\setlit{\setlit{ \stitch{\Xi}{p.2:\tau}}\mid p \gets \Xi(\varphi), p.1 = \tuple{\vect{V}} }
    \\
      \stitch{\Xi}{\kwindex(\varphi,\vect{V}):\msetlit{\tau}}
    & \triangleq &
      \biguplus\msetlit{\msetlit{ \stitch{\Xi}{p.2:\tau}}\mid p \gets \Xi(\varphi), p.1 = \tuple{\vect{V}} }
  \end{array}
    }
    \]
  \caption{The stitching function.}\label{fig:stitching}
  \end{figure}

  The stitching function is shown in Figure~\ref{fig:stitching}: its
  job is to visit all the components of tuples and collections,
  ignoring atomic values other than indices along the way.  The real
  work is performed when an $\kwindex(\varphi,\vect{V})$ is found:
  conceptually, the index should be replaced by the result of the
  evaluation of $\varphi \gapp(\vect{V})$. Remember that $\Xi$
  contains the result of the evaluation of the graph function
  $\varphi$ after translation to \NRClsb, i.e. a collection of pairs
  associating each input of $\varphi$ to the corresponding output:
  then, to obtain the desired result, we can take $\Xi(\varphi)$,
  filter all the pairs $p$ whose first component is
  $\tuple{\vect{V}}$, and return the second component of $p$ after a
  recursive stitching.  Finally, observe that we track the result type
  argument in order to disambiguate whether to construct a set or
  multiset when we encounter an index.

  \begin{theorem}[Correctness of stitching]\label{thm:stitchsound}
    Let $\Theta$ be well-typed and $ty(\Theta) \vdash M : \sigma$.
    Let $\Phi$ be well-typed, and  suppose $\Phi;\Theta \vdash M \shreds \flat{M} \orelse \Psi$. Let $\Xi$ be the result of evaluating
    the flattened queries in $\Psi$ as above.  Then
    $\sem{\flat{M}\Psi}\rho = \sem{\stitch{\Xi}{\defun{\flat{M}}:\tau}}\rho$.
  \end{theorem}

  The full correctness result follows by combining the Theorems~\ref{thm:shredsound} and~\ref{thm:stitchsound}.
  \begin{corollary}
    For all $M$ such that $\vdash M : \tau$, suppose $\vdash M \shreds \flat{M'} \mid \Psi$, and let $\Xi$ be the shredded value set obtained by evaluating the flattened queries in $\Psi$. Then $\sem{M} = \sem{\stitch{\Xi}{\defun{\flat{M}}:\tau}}$.
  \end{corollary}

\section{Related work}\label{sec:related}

Work on language-integrated query and comprehension syntax has taken
place over several decades in both the database and programming
language communities.  We discuss the most closely related work below.

\paragraph{Comprehensions, normalization and language integration}
The database community had already begun in the late 1980s to explore
proposals for so-called \emph{non-first-normal-form} relations in
which collections could be nested inside other
collections~\cite{schek86is}, but following Trinder and Wadler's initial
work connecting database queries with monadic
comprehensions~\cite{trinder-wadler:comprehensions}, query languages
based on these foundations were studied extensively, particularly by
Buneman et al.~\cite{BNTW95,buneman+:comprehensions}.  For our
purposes, Wong's work on query normalization and translation to
SQL~\cite{wong96jcss} is the most important landmark; this work
provided the basis for practical implementations such as Kleisli and
later Links.  Almost as important is the later work by Libkin and
Wong~\cite{LW97}, studying the questions of expressiveness of bag
query languages via a language \BQL that extended basic \NRC with
deduplication and bag difference operators.  They related this
language to \NRC with set semantics extended with aggregation
(count/sum) operations, but did not directly address the question of
normalizing and translating \BQL queries to SQL. Grust and
Scholl~\cite{grust99jiis} were early advocates of the use of
comprehensions mixing set, bag and other monadic collections for query
rewriting and optimization, but did not study normalization or
translatability properties.

Although comprehension-based queries began to be used in
general-purpose programming languages with the advent of Microsoft
LINQ~\cite{meijer:sigmod} and Links~\cite{CLWY06},
Cooper~\cite{Cooper09} made the next important foundational
contribution by extending Wong's normalization result to queries
containing higher-order functions and showing that an effect system
could be used to safely compose queries using higher-order functions
even in an ambient language with side-effects and recursive functions
that cannot be used in queries.  This work provided the basis for
subsequent development of language-integrated query in
Links~\cite{lindley12tldi} and was later adapted for use in
F\#~\cite{cheney13icfp}, Scala~\cite{quill}, and by Kiselyov et
al.~\cite{suzuki16pepm} in the OCaml library \QueL.  However, on
revisiting Cooper's proof to extend it to heterogeneous queries, we
found a subtle gap in the proof, which was corrected in a recent
paper~\cite{ricciotti20fscd}; the original result was correct.  As a
result, in this paper we focus on first-order fragments of these
languages without loss of generality.

Giorgidze et al.~\cite{giorgidze13ddfp} have shown how to support
non-recursive datatypes (i.e. sums) and Grust and
Ulrich~\cite{grust13dbpl} built on this to show how to support
function types in query results using
defunctionalization~\cite{grust13dbpl}.  We considered using sums to
support  a
defunctionalization-style strategy for query lifting, but Giorgidze et
al.~\cite{giorgidze13ddfp} map sum types to nested collections, which
makes their approach unsuitable to our setting.  Wong's original
normalization result also considered sum types, but to the best of our
knowledge normalization for \NRClsb extended with sum types has not
yet been proved.

Recent work by Suzuki et al.~\cite{suzuki16pepm} have outlined further
extensions to lan\-guage-integrated query in the \QueL system, which is
based on finally-tagless syntax~\cite{carette09jfp} and employs Wong's
and Cooper's rewrite rules; Katsushima and Kiselyov's subsequent short
paper~\cite{katsushima17pepm} outlined extensions to handling ordering
and grouping. Kiselyov and Katsushima~\cite{kiselyov17aplas} present
an extension to \QueL called \SQUR to handle ordering based on effect
typing, and they provide an elegant translation from \SQUR queries to
SQL based on normalization-by-evaluation.  Okura and
Kameyama~\cite{okura20flops} outline an extension to handle SQL-style
grouping and aggregation operators in \QueLG; however, their approach
potentially generates lateral variable occurrences inside grouping
queries.  These systems \QueL, \SQUR and \QueLG consider neither
heterogeneity nor nested results.  

Our adoption of tabulated functions (\emph{graphs}) is inspired in
part by Gibbons et al.~\cite{gibbons18icfp}, who provided an elegant
rational reconstruction of relational algebra showing how standard
principles for reasoning about queries arise from adjunctions.  They
employed types for (finite) maps and tables to show how joins can be
implemented efficiently, and observed that such structures form a
\emph{graded monad}.  We are interested in further exploring these
structures and extending our work to cover ordering, grouping and
aggregation.

\paragraph{Query decorrelation and delateralization}
There is a large literature on \emph{query decorrelation}, for example
to remove aggregation operations from $\kwsel$ or $\kwsqlwhere$
clauses (see e.g.~\cite{neumann15btw,cao07tods} for further
discussion).  Delateralization appears related to decorrelation, but
we are aware of only a few works on this problem, perhaps because most
DBMSs only started to support $\kwlateral$ in the last few years.
(Microsoft SQL Server has supported similar functionality for much
longer through a keyword $\kwapply$.)  Our delateralization
technique appears most closely related to Neumann and Kemper's work on
query unnesting~\cite{neumann15btw}.  In this context, unnesting
refers to removal of ``dependent join'' expressions in a relational
algebraic query language; such joins appear to correspond to lateral
subqueries.  This approach is implemented in the HyPER database
system, but is not accompanied by a proof of correctness, nor does it
handle nested query results.  It would be interesting to formalize
this approach (or others from the decorrelation literature) and relate
it to delateralization.  

\paragraph{Querying nested collections}

Our approach to querying nested heterogeneous collections clearly
specializes to the homogeneous cases for sets and multisets
respectively, which have been studied separately.  Van den Bussche's
work on simulating queries on nested sets using flat
ones~\cite{Bussche01} has also inspired subsequent work on query
shredding, flattening and (in this paper) lifting, though the
simulation technique itself does not appear practical (as discussed in
the extended version of Cheney et al.~\cite{cheney14tr}).  More
recently, Benedikt and Pradic~\cite{benedikt21popl} presented results
on representing queries on nested collections using a bounded number
of \emph{interpretations} (first-order logic formulas corresponding to
definable flat query expressions) in the context of their work on
\emph{synthesizing} \NRC queries from proofs.  This approach considers
set-valued \NRC only, and its relationship to our approach should be
investigated further.

Cheney et al.'s previous work on query shredding for multiset
queries~\cite{cheney14sigmod} is different in several important respects.  In
that work we did not consider deduplication and bag difference
operations from \BQL, which Libkin and Wong showed cannot be expressed
in terms of other \NRC operations.  The shredding translation was given
in several stages, and while each stage is individually
comprehensible, the overall approach is not easy to understand.
Finally, the last stages of the translation relied on SQL features not
present (or expressible) in the source language, such as ordering and
the SQL:1999 $\kwrownum$ construct, to synthesize uniform integer
keys.  Our approach, in contrast, handles set, bag, and mixed queries,
and does not rely on any SQL:1999 features.

In a parallel line of work, Grust et
al.~\cite{grust10vldb,dsh,Ulrich11,SIGMOD2015UlrichG,ulrich19phd} have
developed a number of approaches to querying nested \emph{list} data
structures, first in the context of XML processing~\cite{grust10edbt}
and subsequently for \NRC-like languages over lists. The earlier
approach~\cite{grust10vldb}, named \emph{loop-lifting} (not to be
confused with \emph{query lifting}!) made heavy use of SQL:1999
capabilities for numbering and indexing to decouple nested collections
from their context, and was implemented in both Links~\cite{Ulrich11}
and earlier versions of the Database Supported Haskell
library~\cite{dsh}, both of which relied on an advanced query
optimizer called \emph{Pathfinder}~\cite{GrustRT08} to optimize these queries.
The more recent approach, implemented by Ulrich in the current version
of DSH and described in detail in his thesis~\cite{ulrich19phd}, is
called \emph{query flattening} and is instead based on techniques from
\emph{nested data parallelism}~\cite{blelloch90}.  Both loop-lifting
and query flattening are very powerful, and do not rely on an initial
normalization stage, while supporting a rich source language with list
semantics, ordering, grouping, aggregation, and deduplication which
can in principle emulate set or multiset semantics. However, to the
best of our knowledge no correctness proofs exist for either
technique.  We view finding correctness results for richer query
languages as an important challenge for future work.

Another parallel line of work started by Fegaras and
Maier~\cite{DBLP:journals/tods/FegarasM00,fegaras17jfp} considers
heterogeneous query languages based on \emph{monoid} comprehensions,
with set, list, and bag collections as well as grouping, aggregation
and ordering operations, in the setting of object-oriented databases,
and forms the basis for complex object database systems such as
$\lambda$DB~\cite{fegaras00sigmod} and Apache
MRQL~\cite{fegaras17jfp}.  However, Wong-style normalization results
or translations from flat or nested queries to SQL are not known for
these calculi.

\paragraph{Lambda-lifting and closure conversion} Since Johnsson's
original work~\cite{johnsson85fpca},
\linebreak
lambda-lifting and closure
conversion have been studied extensively for functional languages,
with Minamide et al.'s \emph{typed closure
  conversion}~\cite{minamide96popl} of particular interest in compilers
employing typed intermediate languages.  We plan to study whether
known optimizations in the lambda-lifting and closure conversion
literature offer advantages for query lifting.
The immediate important next step is to implement our approach and
compare it empirically with previous techniques such as query
shredding and query flattening.  By analogy with lambda-lifting and
closure conversion, we expect additional optimizations to be possible
by a deeper analysis of how variables/fields are used in lifted
subqueries.  Another problem we have not resolved is how to deal with
deduplication or bag difference at nested collection types in
practice.  Libkin and Wong~\cite{LW97} showed that such nesting can be
eliminated from \BQL queries, but their results do not provide a
constructive algorithm for eliminating the nesting.

\section{Conclusions}\label{sec:concl}

Monadic comprehensions have proved to be a remarkably durable
foundation for database programming and language-integrated query, and
has led to language support (LINQ for .NET, Quill for Scala) with
widespread adoption.  Recent work has demonstrated that techniques for
evaluating queries over nested collections, such as query shredding or
query flattening, can offer order-of-magnitude speedups in database
applications~\cite{fowler20ijdc} without sacrificing declarativity or
readability.  However, query shredding lacks the ability to express
common operations such as deduplication, while query flattening is
more expressive but lacks a detailed proof of correctness, and both
techniques are challenging to understand, implement, or extend.  We provide
the first provably correct approach to querying nested heterogeneous
collections involving both sets and multisets.

Our most important insight is that working in a heterogeneous
language, with both set and multiset collection types, actually makes
the problem easier, by making it possible to calculate finite maps
representing the behavior of nested query subexpressions under all of
the possible environments encountered at run time.  Thus, instead of
having to maintain or synthesize keys linking inner and
outer collections, as is done in all previous approaches, we can
instead use the values of variables in the closures of nested query
expressions themselves as the keys.  The same approach can be used to
eliminate sideways information-passing.  This is analogous to
lambda-lifting or closure conversion in compilation of functional
languages, but differs in that we lift local queries to (queries that
compute) finite maps rather than ordinary function abstractions.  We
believe this idea may have broader applications and will next
investigate its behavior in practice and applications to other query
language features.

\paragraph{Acknowledgments}
This work was supported by ERC Consolidator Grant Skye (grant number
682315), and by an ISCF Metrology Fellowship grant provided by the UK
government's Department for Business, Energy and Industrial Strategy
(BEIS). We are grateful to Simon Fowler for feedback and to anonymous
reviewers for constructive comments.

\bibliography{paper}

\begin{thebibliography}{10}
\providecommand{\url}[1]{\texttt{#1}}
\providecommand{\urlprefix}{URL }
\providecommand{\doi}[1]{https://doi.org/#1}

\bibitem{benedikt21popl}
Benedikt, M., Pradic, P.: Generating collection transformations from proofs.
  Proc. ACM Program. Lang.  \textbf{5}(POPL) (Jan 2021),
  \url{https://doi.org/10.1145/3434295}

\bibitem{blelloch90}
Blelloch, G.E.: Vector Models for Data-Parallel Computing. {MIT} Press (1990)

\bibitem{buneman+:comprehensions}
Buneman, P., Libkin, L., Suciu, D., Tannen, V., Wong, L.: Comprehension syntax.
  SIGMOD Record  \textbf{23} (1994)

\bibitem{BNTW95}
Buneman, P., Naqvi, S., Tannen, V., Wong, L.: Principles of programming with
  complex objects and collection types. Theor. Comput. Sci.  \textbf{149}(1)
  (1995). \doi{10.1016/0304-3975(95)00024-Q}

\bibitem{cao07tods}
Cao, B., Badia, A.: {SQL} query optimization through nested relational algebra.
  ACM Trans. Database Syst.  \textbf{32}(3),  18–es (Aug 2007).
  \doi{10.1145/1272743.1272748}

\bibitem{carette09jfp}
Carette, J., Kiselyov, O., Shan, C.: Finally tagless, partially evaluated:
  Tagless staged interpreters for simpler typed languages. J. Funct. Program.
  \textbf{19}(5),  509--543 (2009). \doi{10.1017/S0956796809007205}

\bibitem{cheney13icfp}
Cheney, J., Lindley, S., Wadler, P.: A practical theory of language-integrated
  query. In: ICFP (2013). \doi{10.1145/2500365.2500586}

\bibitem{cheney14sigmod}
Cheney, J., Lindley, S., Wadler, P.: Query shredding: efficient relational
  evaluation of queries over nested multisets. In: SIGMOD. pp. 1027--1038. ACM
  (2014). \doi{10.1145/2588555.2612186}

\bibitem{cheney14tr}
Cheney, J., Lindley, S., Wadler, P.: Query shredding: Efficient relational
  evaluation of queries over nested multisets (extended version). CoRR
  \textbf{abs/1404.7078} (2014), \url{http://arxiv.org/abs/1404.7078}

\bibitem{Chu17}
Chu, S., Weitz, K., Cheung, A., Suciu, D.: {HoTTSQL}: Proving query rewrites
  with univalent {SQL} semantics. In: PLDI. pp. 510--524. ACM (2017).
  \doi{10.1145/3062341.3062348}

\bibitem{Cooper09}
Cooper, E.: The script-writer's dream: How to write great {SQL} in your own
  language, and be sure it will succeed. In: DBPL (2009).
  \doi{10.1007/978-3-642-03793-1\_3}

\bibitem{CLWY06}
Cooper, E., Lindley, S., Wadler, P., Yallop, J.: Links: web programming without
  tiers. In: FMCO (2007). \doi{10.1007/978-3-540-74792-5\_12}

\bibitem{copeland-maier:1984}
Copeland, G., Maier, D.: Making {S}malltalk a database system. SIGMOD Rec.
  \textbf{14}(2) (1984)

\bibitem{fegaras17jfp}
Fegaras, L.: An algebra for distributed big data analytics. J. Funct. Program.
  \textbf{27}, ~e27 (2017). \doi{10.1017/S0956796817000193}

\bibitem{DBLP:journals/tods/FegarasM00}
Fegaras, L., Maier, D.: Optimizing object queries using an effective calculus.
  ACM Trans. Database Syst.  \textbf{25}(4),  457--516 (2000)

\bibitem{fegaras00sigmod}
Fegaras, L., Srinivasan, C., Rajendran, A., Maier, D.: lambda-{DB}: An
  {ODMG}-based object-oriented {DBMS}. In: Chen, W., Naughton, J.F., Bernstein,
  P.A. (eds.) SIGMOD. p.~583. {ACM} (2000). \doi{10.1145/342009.335494}

\bibitem{fehrenbach18scp}
Fehrenbach, S., Cheney, J.: Language-integrated provenance. Science of Computer
  Programming  \textbf{155},  103--145 (2018)

\bibitem{foster08pods}
Foster, J.N., Green, T.J., Tannen, V.: Annotated {XML}: queries and provenance.
  In: PODS. pp. 271--280 (2008)

\bibitem{fowler20ijdc}
Fowler, S., Harding, S., Sharman, J., Cheney, J.: Cross-tier web programming
  for curated databases: a case study. International Journal of Digital
  Curation  \textbf{15}(1) (2020). \doi{10.2218/ijdc.v15i1.717}, pre-print
  presented at IDCC 2020

\bibitem{gibbons18icfp}
Gibbons, J., Henglein, F., Hinze, R., Wu, N.: Relational algebra by way of
  adjunctions. Proc. ACM Program. Lang.  \textbf{2}(ICFP) (Jul 2018).
  \doi{10.1145/3236781}

\bibitem{dsh}
Giorgidze, G., Grust, T., Schreiber, T., Weijers, J.: {Haskell} boards the
  {Ferry} - database-supported program execution for {Haskell}. In: IFL. pp.
  1--18. No.~6647 in LNCS, Springer-Verlag (2010)

\bibitem{giorgidze13ddfp}
Giorgidze, G., Grust, T., Ulrich, A., Weijers, J.: Algebraic data types for
  language-integrated queries. In: DDFP. pp. 5--10 (2013)

\bibitem{green07pods}
Green, T.J., Karvounarakis, G., Tannen, V.: Provenance semirings. In: PODS
  (2007)

\bibitem{grust10edbt}
Grust, T., Mayr, M., Rittinger, J.: Let {SQL} drive the {XQuery} workhorse
  ({XQuery} join graph isolation). In: EDBT. pp. 147--158 (2010).
  \doi{10.1145/1739041.1739062}

\bibitem{GrustMRS09}
Grust, T., Mayr, M., Rittinger, J., Schreiber, T.: Ferry: Database-supported
  program execution. In: SIGMOD (June 2009)

\bibitem{grust10vldb}
Grust, T., Rittinger, J., Schreiber, T.: Avalanche-safe {LINQ} compilation.
  PVLDB  \textbf{3}(1) (2010)

\bibitem{GrustRT08}
Grust, T., Rittinger, J., Teubner, J.: Pathfinder: {XQuery} off the relational
  shelf. IEEE Data Eng. Bull.  \textbf{31}(4) (2008)

\bibitem{grust99jiis}
Grust, T., Scholl, M.H.: How to comprehend queries functionally. J. Intell.
  Inf. Syst.  \textbf{12}(2-3),  191--218 (1999). \doi{10.1023/A:1008705026446}

\bibitem{grust13dbpl}
Grust, T., Ulrich, A.: First-class functions for first-order database engines.
  In: DBPL (2013), \texttt{http://arxiv.org/abs/1308.0158}

\bibitem{johnsson85fpca}
Johnsson, T.: Lambda lifting: Treansforming programs to recursive equations.
  In: FPCA. pp. 190--203 (1985). \doi{10.1007/3-540-15975-4\_37}

\bibitem{katsushima17pepm}
Katsushima, T., Kiselyov, O.: Language-integrated query with ordering, grouping
  and outer joins (poster paper). In: PEPM. pp. 123--124 (2017)

\bibitem{kiselyov17aplas}
Kiselyov, O., Katsushima, T.: Sound and efficient language-integrated query -
  maintaining the {ORDER}. In: APLAS 2017. pp. 364--383 (2017).
  \doi{10.1007/978-3-319-71237-6\_18}

\bibitem{LW97}
Libkin, L., Wong, L.: Query languages for bags and aggregate functions. J.
  Comput. Syst. Sci.  \textbf{55}(2) (1997). \doi{10.1006/jcss.1997.1523}

\bibitem{lindley12tldi}
Lindley, S., Cheney, J.: Row-based effect types for database integration. In:
  TLDI (2012). \doi{10.1145/2103786.2103798}

\bibitem{lindley10esop}
Lindley, S., Wadler, P.: The audacity of hope: Thoughts on reclaiming the
  database dream. In: ESOP (2010)

\bibitem{meijer:sigmod}
Meijer, E., Beckman, B., Bierman, G.M.: {LINQ}: reconciling object, relations
  and {XML} in the {.NET} framework. In: SIGMOD (2006).
  \doi{10.1145/1142473.1142552}

\bibitem{minamide96popl}
Minamide, Y., Morrisett, J.G., Harper, R.: Typed closure conversion. In: POPL.
  pp. 271--283 (1996). \doi{10.1145/237721.237791}

\bibitem{neumann15btw}
Neumann, T., Kemper, A.: Unnesting arbitrary queries. In: Datenbanksysteme
  f{\"{u}}r Business, Technologie und Web (BTW). pp. 383--402 (2015)

\bibitem{okura20flops}
Okura, R., Kameyama, Y.: Language-integrated query with nested data structures
  and grouping. In: FLOPS. pp. 139--158 (2020).
  \doi{10.1007/978-3-030-59025-3\_9}

\bibitem{ParedaensG92}
Paredaens, J., {Van Gucht}, D.: Converting nested algebra expressions into flat
  algebra expressions. ACM Trans. Database Syst.  \textbf{17}(1) (1992).
  \doi{10.1145/128765.128768}

\bibitem{quill}
Quill: Compile-time language integrated queries for {Scala}. Open source
  project, https://github.com/getquill/quill

\bibitem{ricciotti19dbpl}
Ricciotti, W., Cheney, J.: Mixing set and bag semantics. In: DBPL. pp. 70--73
  (2019). \doi{10.1145/3315507.3330202}

\bibitem{ricciotti20fscd}
Ricciotti, W., Cheney, J.: Strongly normalizing higher-order relational
  queries. In: FSCD. pp. 28:1--28:22 (2020). \doi{10.4230/LIPIcs.FSCD.2020.28}

\bibitem{russell08queue}
Russell, C.: Bridging the object-relational divide. Queue  \textbf{6} (May
  2008). \doi{10.1145/1394127.1394139}

\bibitem{schek86is}
Schek, H., Scholl, M.H.: The relational model with relation-valued attributes.
  Inf. Syst.  \textbf{11}(2),  137--147 (1986).
  \doi{10.1016/0306-4379(86)90003-7}

\bibitem{stolarek18programming}
Stolarek, J., Cheney, J.: Language-integrated provenance in {Haskell}. The Art,
  Science, and Engineering of Programming  \textbf{2}(3), ~A11 (2018)

\bibitem{suzuki16pepm}
Suzuki, K., Kiselyov, O., Kameyama, Y.: Finally, safely-extensible and
  efficient language-integrated query. In: PEPM. pp. 37--48 (2016).
  \doi{10.1145/2847538.2847542}

\bibitem{Syme06}
Syme, D.: Leveraging {.NET} meta-programming components from {F\#}: integrated
  queries and interoperable heterogeneous execution. In: ML Workshop (2006)

\bibitem{trinder-wadler:comprehensions}
Trinder, P., Wadler, P.: Improving list comprehension database queries. In:
  TENCON '89. (1989). \doi{10.1109/TENCON.1989.176921}

\bibitem{Ulrich11}
Ulrich, A.: A {Ferry}-based query backend for the {Links} programming language.
  Master's thesis, University of T\"ubingen (2011)

\bibitem{ulrich19phd}
Ulrich, A.: Query Flattening and the Nested Data Parallelism Paradigm. Ph.D.
  thesis, University of T{\"{u}}bingen, Germany (2019)

\bibitem{SIGMOD2015UlrichG}
Ulrich, A., Grust, T.: The flatter, the better: Query compilation based on the
  flattening transformation. In: SIGMOD. pp. 1421--1426. ACM (2015).
  \doi{10.1145/2723372.2735359}

\bibitem{Bussche01}
{Van den Bussche}, J.: Simulation of the nested relational algebra by the flat
  relational algebra, with an application to the complexity of evaluating
  powerset algebra expressions. Theor. Comput. Sci.  \textbf{254}(1-2) (2001)

\bibitem{wong96jcss}
Wong, L.: Normal forms and conservative extension properties for query
  languages over collection types. J. Comput. Syst. Sci.  \textbf{52}(3)
  (1996). \doi{10.1006/jcss.1996.0037}

\bibitem{Won00}
Wong, L.: Kleisli, a functional query system. J. Funct. Program.
  \textbf{10}(1) (2000). \doi{10.1017/S0956796899003585}

\end{thebibliography}

\begin{techreport}
\newpage
\appendix
\section{\NRClsb}

\subsection{Type system}
We give here the full set of typing rules for \NRClsb that we omitted from the main body of the paper: they are shown in Figure~\ref{fig:typing}.

\begin{figure}
  \begin{center}
  \AxiomC{$x : \sigma \in \Gamma$}
  \UnaryInfC{$\Gamma \vdash x : \sigma$}
  \DisplayProof
  \hspace{.5cm}
  \AxiomC{$\Sigma(c) = \vect{b} \to b$}
  \AxiomC{$(\Gamma \vdash M_i : b_i)_{i = 1,\ldots,n}$}
  \BinaryInfC{$\Gamma \vdash c(\vect{M_n}) : b$}
  \DisplayProof
  
  \medskip
  
  \AxiomC{$(\Gamma \vdash M_i : \sigma_i)_{i = 1, \ldots, n}$}
  \UnaryInfC{$\Gamma \vdash \tuple{\vect{\ell = M}} : \tuple{\vect{\ell : \sigma}}$}
  \DisplayProof
  \hspace{.5cm}
  \AxiomC{$\Gamma \vdash M : \tuple{\vect{\ell : \sigma}}$}
  \AxiomC{$i = 1,\ldots,n$}
  \BinaryInfC{$\Gamma \vdash M.\ell_i : \sigma_i$}
  \DisplayProof
  
  \medskip
  
  \AxiomC{$\Gamma, x : \sigma \vdash M : \tau$}
  \UnaryInfC{$\Gamma \vdash \lambda x^\sigma.M : \sigma \to \tau$}
  \DisplayProof
  \hspace{.5cm}
  \AxiomC{$\Gamma \vdash M : \sigma \to \tau$}
  \AxiomC{$\Gamma \vdash N : \sigma$}
  \BinaryInfC{$\Gamma \vdash (M~N) : \tau$}
  \DisplayProof
  
  \medskip
  
  \AxiomC{$\phantom{A}$}
  \UnaryInfC{$\Gamma \vdash \emptyoset^\sigma : \setlit{\sigma}$}
  \DisplayProof
  \hspace{.5cm}
  \AxiomC{$\Gamma \vdash M : \sigma$}
  \UnaryInfC{$\Gamma \vdash \setlit{M} : \setlit{\sigma}$}
  \DisplayProof
  \hspace{.5cm}
  \AxiomC{$\Gamma \vdash M : \setlit{\sigma}$}
  \AxiomC{$\Gamma \vdash N : \setlit{\sigma}$}
  \BinaryInfC{$\Gamma \vdash M \ocup N : \setlit{\sigma}$}
  \DisplayProof
  
  \medskip
  
  \AxiomC{$(\Gamma, x_1 : \sigma_1,\ldots,x_{i-1}:\sigma_{i-1} \vdash N_i : \setlit{\sigma_i})_{i=1,\ldots,n}$}
  \AxiomC{$\Gamma, \vect{x:\sigma} \vdash M : \setlit{\tau}$}
  \BinaryInfC{$\Gamma \vdash \comprehension{M | \vect{x \leftarrow N}} : \setlit{\tau}$}
  \DisplayProof
  
  \medskip
  
  \AxiomC{$\Gamma \vdash M : \setlit{\sigma}$}
  \UnaryInfC{$\Gamma \vdash \setempty(M) : \boolty$}
  \DisplayProof
  \hspace{.5cm}
  \AxiomC{$\Gamma \vdash M : \setlit{\sigma}$}
  \AxiomC{$\Gamma \vdash N : \boolty$}
  \BinaryInfC{$\Gamma \vdash M~\setwhere~N : \setlit{\sigma}$}
  \DisplayProof
  
  \medskip
  
  \AxiomC{$\phantom{A}$}
  \UnaryInfC{$\Gamma \vdash \emptymset^\sigma : \msetlit{\sigma}$}
  \DisplayProof
  \hspace{.5cm}
  \AxiomC{$\Gamma \vdash M : \sigma$}
  \UnaryInfC{$\Gamma \vdash \msetlit{M} : \setlit{\sigma}$}
  \DisplayProof
  \hspace{.5cm}
  \AxiomC{$\Gamma \vdash M : \msetlit{\sigma}$}
  \AxiomC{$\Gamma \vdash N : \msetlit{\sigma}$}
  \BinaryInfC{$\Gamma \vdash M \uplus N : \msetlit{\sigma}$}
  \DisplayProof
  
  \medskip
  
  \AxiomC{$(\Gamma, x_1 : \sigma_1,\ldots,x_{i-1}:\sigma_{i-1} \vdash N_i : \msetlit{\sigma_i})_{i=1,\ldots,n}$}
  \AxiomC{$\Gamma, \vect{x:\sigma} \vdash M : \msetlit{\tau}$}
  \BinaryInfC{$\Gamma \vdash \mcomprehension{M | \vect{x \leftarrow N}} : \msetlit{\tau}$}
  \DisplayProof
  
  \medskip
  
  \AxiomC{$\Gamma \vdash M : \msetlit{\sigma}$}
  \UnaryInfC{$\Gamma \vdash \bagempty(M) : \boolty$}
  \DisplayProof
  \hspace{.5cm}
  \AxiomC{$\Gamma \vdash M : \msetlit{\sigma}$}
  \AxiomC{$\Gamma \vdash N : \boolty$}
  \BinaryInfC{$\Gamma \vdash M~\bagwhere~N : \msetlit{\sigma}$}
  \DisplayProof

  \medskip

  \AxiomC{$\Gamma \vdash M : \msetlit{\sigma}$}
  \UnaryInfC{$\Gamma \vdash \distinct M : \setlit{\sigma}$}
  \DisplayProof
  \hspace{.5cm}
  \AxiomC{$\Gamma \vdash M : \setlit{\sigma}$}
  \UnaryInfC{$\Gamma \vdash \promote M : \msetlit{\sigma}$}
  \DisplayProof
  \end{center}

  \caption{Type system of \NRClsb.}\label{fig:typing}
  \end{figure}

\subsection{Normalization}
We show in Figure~\ref{fig:normNRClsb} the rewrite system used to normalize \NRClsb queries.

\begin{figure}
  \begin{center}\small
  \begin{tabular}{r@{$~\red~$}l}
  \multicolumn{2}{c}{$
  (\lambda x.M)~N \red M[N/x]
  \qquad
  \tuple{\ldots, \ell = M, \ldots}.\ell \red M
  $}
  \\
  \multicolumn{2}{c}{} 
  \\
  \multicolumn{2}{c}{$
  \comprehension{\emptyoset | \Theta} \red \emptyoset
  \qquad
  \comprehension{M | \Theta, x \leftarrow \emptyoset,\Theta'} \red \emptyoset
  $}
  \\
  $\comprehension{M | \Theta, x \leftarrow \setlit{N}, \Theta'}$ 
  & $\comprehension{M\subst{x}{N} | \Theta, \Theta'\subst{x}{N}}$
  \\
  $\comprehension{M \ocup N | \Theta}$ 
   & $\comprehension{M | \Theta} \ocup \comprehension{N | \Theta}$
  \\
  $\comprehension{M | \Theta,x \leftarrow N \ocup R,\Theta'}$
  & $\comprehension{M | \Theta,x \leftarrow N,\Theta'} \ocup \comprehension{M | \Theta,x \leftarrow R,\Theta'}$
  \\
  $\comprehension{M | \Theta, x \leftarrow \comprehension{R | \Theta' }, \Theta''}$
  & $\comprehension{ M | \Theta, \Theta', x \leftarrow R, \Theta''}$ (if $\dom(\Theta') \notin \FV(M,\Theta'')$)
  \\
  $\comprehension{\comprehension{M | \Theta'} | \Theta}$
  & $\comprehension{M | \Theta,\Theta'}$
  \\
  $\comprehension{M | \Theta,x \leftarrow R~\setwhere~N,\Theta'}$ 
  & $\comprehension{M~\setwhere~N | \Theta, x \leftarrow R, \Theta'}$ (if $x \notin \FV(N)$)
  \\
  $\comprehension{\distinct(M - N) | \Theta}$
  &
  $\comprehension{\setlit{z} | \Theta, z \gets \distinct(M - N)}$
  \\
  $\comprehension{\distinct(M - N)~\setwhere~R | \Theta}$
  &
  $\comprehension{\setlit{z}~\setwhere~R | \Theta, z \gets \distinct(M - N)}$ (if $z \notin \FV(R)$)
  \\
  \multicolumn{2}{c}{} 
  \\
  \multicolumn{2}{c}{$
  M~\setwhere~\kwtrue \red M
  \qquad
  M~\setwhere~\kwfalse \red \emptyoset
  \qquad
  \emptyoset~\setwhere~M \red \emptyoset
  $}
  \\
  $(N \ocup R)~\setwhere~M$ & $(N~\setwhere~M) \ocup (R~\setwhere~M)$
  \\
  $\comprehension{N | \Theta}~\setwhere~M$ 
  & $\comprehension{N~\setwhere~M | \Theta}$ (if $\dom(\Theta) \cap FV(M) = \emptyset$)
  \\
  $R~\setwhere~N~\setwhere~M$ & $R~\setwhere~(M \land N)$
  \\
  \multicolumn{2}{c}{} 
  \\
  \multicolumn{2}{c}{$
  \mcomprehension{\emptymset | \Theta} \red \emptymset
  \qquad
  \mcomprehension{M | \Theta, x \leftarrow \emptymset,\Theta'} \red \emptymset
  $}
  \\
  $\mcomprehension{M | \Theta, x \leftarrow \msetlit{N}, \Theta'}$ 
  & $\mcomprehension{M\subst{x}{N} | \Theta, \Theta'\subst{x}{N}}$
  \\
  $\mcomprehension{M \uplus N | \Theta}$ 
   & $\mcomprehension{M | \Theta} \uplus \mcomprehension{N | \Theta}$
  \\
  $\mcomprehension{M | \Theta,x \leftarrow N \uplus R,\Theta'}$
  & $\mcomprehension{M | \Theta,x \leftarrow N,\Theta'} \uplus \mcomprehension{M | \Theta,x \leftarrow R,\Theta'}$
  \\
  $\mcomprehension{M | \Theta, x \leftarrow \mcomprehension{R | \Theta' }, \Theta''}$
  & $\mcomprehension{ M | \Theta, \Theta', x \leftarrow R, \Theta''}$ (if $\dom(\Theta') \notin \FV(M,\Theta'')$)
  \\
  $\mcomprehension{\mcomprehension{M | \Theta'} | \Theta}$
  & $\mcomprehension{M | \Theta,\Theta'}$
  \\
  $\mcomprehension{M | \Theta, x \leftarrow R~\bagwhere~N, \Theta'}$ 
  & $\mcomprehension{M~\bagwhere~N | \Theta, x \leftarrow R, \Theta'}$ (if $x \notin \FV(N)$)
  \\
  $\mcomprehension{\promote M | \Theta}$
  & $\mcomprehension{\msetlit{z} | \Theta, z \gets \promote M}$
  \\
  $\mcomprehension{\promote M~\bagwhere~R | \Theta}$
  &
  $\mcomprehension{\msetlit{z}~\bagwhere~R | \Theta, z \gets \promote M}$ (if $z \notin \FV(R)$)
  \\
  $\mcomprehension{ M - N | \Theta}$
  & $\mcomprehension{\msetlit{z} | \Theta, z \gets M - N}$
  \\
  $\mcomprehension{(M - N)~\bagwhere~R | \Theta}$
  &
  $\mcomprehension{\msetlit{z}~\bagwhere~R | \Theta, z \gets M - B}$ (if $z \notin \FV(R)$)
  \\
  \multicolumn{2}{c}{} 
  \\
  \multicolumn{2}{c}{$
  M~\bagwhere~\kwtrue \red M
  \qquad
  M~\bagwhere~\kwfalse \red \emptymset
  \qquad
  \emptymset~\bagwhere~M \red \emptymset
  $}
  \\
  $(N \uplus R)~\bagwhere~M$ & $(N~\bagwhere~M) \uplus (R~\bagwhere~M)$
  \\
  $\mcomprehension{N | \Theta}~\bagwhere~M$ 
  & $\mcomprehension{N~\bagwhere~M | \Theta}$ (if $\dom(\Theta) \cap \FV(M) = \emptyset$)
  \\
  $R~\bagwhere~N~\bagwhere~M$ & $R~\bagwhere~(M \land N)$
  \\
  \multicolumn{2}{c}{} 
  \\
  \multicolumn{2}{c}{$
    \distinct\emptymset \red \emptyset
    \qquad
    \distinct\msetlit{M} \red \setlit{M}
    \qquad
    \distinct (M \uplus N) \red \distinct M \cup \distinct N
    $}
  \\
  \multicolumn{2}{c}{$
      \distinct\mcomprehension{M | \Theta} \red \comprehension{\distinct M | \Theta^\distinct}
      \qquad
      \distinct(M~\bagwhere~N) \red \distinct M~\setwhere~N
  $}
      \\
  \multicolumn{2}{c}{$
      \distinct\promote M \red M
      \qquad
      \promote \emptyset \red \emptymset
      \qquad
      \promote(M~\setwhere~N) \red \promote M~\bagwhere~N
  $}
  \\
  $\setempty(M)$ & $\setempty(\comprehension{\setlit{\tuple{}} | x \gets M})$ (if $M$ is not a flat set)
  \\
  $\bagempty(M)$ & $\bagempty(\comprehension{\msetlit{\tuple{}} | x \gets M})$ (if $M$ is not a flat bag)
  \end{tabular}

  \[
  \begin{array}{rl}
  (\vect{x \gets M})\subst{y}{N} & \triangleq \vect{x \gets M\subst{y}{N}}
  \qquad \text{(if $x \neq y, \FV(N) \cap \vect{x} = \emptyset$)}
  \\
  (\vect{x \gets M})^\distinct & \triangleq \vect{x \gets \distinct M}  
  \end{array}
  \]
  \end{center}
  \caption{Query normalization}\label{fig:normNRClsb}
  \end{figure}  

\subsection{Semantics}
We follow the $K$-relation style of semantics, as introduced by 
Green et al.~\cite{green07pods} and used for formalization by Chu et al.~\cite{Chu17}.

Basic types and records are represented by the usual interpretations
of such types, and the details are elided.  For set types, the
interpretation of a set $\setlit{A}$ is $\sem{A} \to_{\mathsf{fs}} \{0,1\}$.
Here $\to_{\mathsf{fs}}$ is the set of finitely-supported functions from
$\sem{A}$, here the support is the set of elements mapped to a nonzero
value.  We consider $\{0,1\}$ equipped with the usual structure of a
Boolean algebra, with operations $\wedge,\vee,\neg$, and we consider
equality and other meta-level predicates as functions returning
Boolean values.   Likewise, we
consider bag types $\msetlit{A}$ to be interpreted as
finitely-supported functions $\sem{A} \to_{\mathsf{fs}} \mathbb{N}$, where $\mathbb{N}$ is the set of natural
numbers, equipped with the usual arithmetic operations $+,\monus,\times$;
here $\monus$ is truncated subtraction $m \monus n = \max(m-n,0)$.

\begin{figure}[tb]
\begin{eqnarray*}
  \sem{\emptyset}\rho &=& \lambda u. 0
\\
\sem{\setlit{M}}\rho &=& \lambda u. \setlit{M}\rho = u
\\
\sem{M \cup N}\rho &=& \lambda u. \sem{M}\rho u \vee \sem{N}\rho u
\\
\sem{\bigcup\setlit{N \mid x \leftarrow M}}\rho &=& \lambda u. 
\bigvee_v \sem{M}\rho v \wedge \sem{N}\rho[x \mapsto v] u
\\
\sem{\setempty(M)}\rho &=& \lambda u. \neg(\bigvee_u \sem{M}\rho u)
\\
\sem{M~\setwhere~ N}\rho &=& \lambda u. \sem{M}\rho u \wedge \sem{N}\rho
\\
\sem{\delta(M)}\rho &=& \lambda u. \zeta(\sem{M}\rho u)
\\
  \sem{\emptymset}\rho &=& \lambda u. 0
\\
\sem{\msetlit{M}}\rho &=& \lambda u. \chi(\setlit{M}\rho = u)
\\
\sem{M \cup N}\rho &=& \lambda u. \sem{M}\rho u + \sem{N}\rho u
\\
\sem{\bigcup\setlit{N \mid x \leftarrow M}}\rho &=& \lambda u. 
\sum_v \sem{M}\rho v \times \sem{N}\rho[x \mapsto v] u
\\
\sem{M~\bagwhere~ N}\rho &=& \lambda u. \sem{M}\rho u \times \chi(\sem{N}\rho)
\\
\sem{\bagempty(M)}\rho &=& \lambda u. \neg(\zeta(\sum_u \sem{M}\rho u))
\\
\sem{M - N}\rho &=& \lambda u. \sem{M}\rho u - \sem{N}\rho u
\\
\sem{\promote(M)}\rho &=& \lambda u. \chi(\sem{M}\rho u)
\end{eqnarray*}
\caption{Semantics of set and multiset operations of
  \NRClsb}\label{fig:semantics}
\end{figure}

Finally to be explicit about the situations where we coerce a Boolean
value to a natural number or vice versa we introduce notation $\chi :
\{0,1\} \to \mathbb{N}$ for the ``characteristic function'' and $\zeta
: \mathbb{N} \to \{0,1\}$ for the ``nonzero test'' function $x \mapsto (x > 0)$ that maps $0$ to $0$ and
any nonzero value to $1$.  Note that $\zeta(\chi(n)) = n$.

Since we work with finitely-supported functions $f,p$, we write $\sum_u
f(u)$ (resp. $\bigvee_u p(u)$ for the summation (resp. disjunction)
over all possible $u$ of $f(u)$ (resp. $p(u)$).  Although this
summation or disjunction is infinite, the number of values of $u$ for
which $f$/$p$ can be nonzero is finite, so this is a finite sum or
disjunction and thus well-defined.  Finally, although \NRClsb also
includes function types, lambda abstraction, and application, but not
recursion, their addition poses no difficulty and since these features can be
normalized away prior to applying the results in this paper, we do not
explicitly discuss them in the semantics.

\section{Proofs for Section~\ref{sec:delat}}

\begin{lemma}
\begin{enumerate}
\item$\sum_t \chi(t=u) \times e(t,u) = e(t,t)$
\item $\chi(\zeta(e)) \times e  = \chi(e > 0) \times e = e$
\end{enumerate}
\end{lemma}
\begin{proof}
For part (1), all of the summands are zero except (possibly) when $t =
u$.  Part (2) follows by a simple case analysis on $e > 0$; if $e = 0$
then both sides are zero while if $e > 0$ then  $\chi(e > 0) \times e
= 1 \times e = e$.\qed
\end{proof}
\begin{lemma}[Commutativity]
Suppose $\{x,y\} \cap FV(M,N) = \emptyset$.  Then
  \[
  \biguplus\msetlit{M \mid x \leftarrow N, y \leftarrow P} \equiv 
  \biguplus\msetlit{M \mid y \leftarrow P, x \leftarrow N}
  \]
  \[
  \bigcup\setlit{M \mid x \leftarrow N, y \leftarrow P} \equiv 
  \bigcup\setlit{M \mid y \leftarrow P, x \leftarrow N}
  \]
\end{lemma}
\begin{proof}
  Straightforward by unfolding definitions.
\qed\end{proof}

Recall (for example from Buneman et al.~\cite{BNTW95}) that set membership $M \in N$ is definable as
$\neg\setempty(\setlit{x \mid x \leftarrow N, x = M})$ It is
straightforward to show that
$\sem{M \in N}\rho v = \sem{N}\rho(\sem{M}\rho)$, that is, the result
is true iff the interpretation of $N$ returns true on the
interpretation of $M$.  We will use
this as a primitive in the following proofs.  First we observe that
when $x$ was introduced by a generator  $x \leftarrow N$, then it
is redundant to check that $x \in \distinct(N)$ (if $N$ is a bag) or
$x \in N$ (if $N$ is a set).

\begin{lemma}\label{lem:graph-elem}
  \[\biguplus\msetlit{M~\bagwhere ~x \in \distinct(N) \mid x \leftarrow N} \equiv
  \biguplus\msetlit{M \mid x \leftarrow N}\]
  \[\bigcup\setlit{M~\setwhere ~x \in N \mid x \leftarrow N} \equiv
  \bigcup\setlit{M \mid x \leftarrow N}\]

\end{lemma}
\begin{proof}
For the first equation we reason as follows:
\begin{eqnarray*}
&&  \sem{\biguplus\msetlit{M~\bagwhere ~x \in \distinct(N) \mid
      x \leftarrow N}}\rho u \\
&=& 
\sum_u\sem{M~\bagwhere ~x \in \distinct(N)}\rho[x\mapsto u]
\times \sem{N}\rho u \\
&=& 
\sum_u\sem{M}\rho[x\mapsto u] v \times \chi(\sem{x \in \distinct(N)}\rho[x\mapsto u])
\times \sem{N}\rho u \\
&=& 
\sum_u\sem{M}\rho[x\mapsto u] v \times
    \chi(\sem{\distinct(N)}\rho[x\mapsto u](\sem{x}\rho[x\mapsto u]))
\times \sem{N}\rho u\\
&=& 
\sum_u\sem{M}\rho[x\mapsto u] v \times
    \chi(\zeta(\sem{N}\rho u))
\times \sem{N}\rho u \\
&=& 
\sum_u \sem{M}\rho[x\mapsto u] v\times \sem{N}\rho u\\
&=&  
\sem{  \biguplus\msetlit{M \mid x \leftarrow N}}\rho v
\end{eqnarray*}
The proof of the second equation is similar, but simpler.
\qed\end{proof}

\begin{lemma}\label{lem:graph-eta}
  \begin{enumerate}
  \item
    $\baggraph{x \leftarrow N}{M}\gapp O \equiv M[O/x] ~\bagwhere~ O
    \in N$
    \item
    $\setgraph{x \leftarrow N}{M}\gapp O \equiv M[O/x] ~\setwhere~ O
    \in N$
  \end{enumerate}
\end{lemma}
\begin{proof}
The proofs are similar; we show the first.
  \begin{eqnarray*}
&&\sem{\graph{x \leftarrow N}{M}\gapp O}\rho v\\
&=&
\sem{\graph{x \leftarrow N}{M}}\rho(\sem{O}\rho,v)\\
&=&
\chi(\sem{N}\rho (\sem{O}\rho)) \times \sem{M}\rho[x\mapsto \sem{O}\rho] v\\
&=&
\sem{M[O/x]}\rho v \times \sem{O \in N}\rho v
\\
&=&
\sem{M[O/x]~ \bagwhere ~O \in N}\rho v
\end{eqnarray*}
\qed\end{proof}

Next we show graph construction commutes with promotion,
deduplication, union, multiset union and difference:
\begin{lemma}
  $\promote(\setgraph{\vect{x \leftarrow N}}{M}) \equiv \baggraph{\vect{x \leftarrow N}}{\promote(M)}$
\end{lemma}
\begin{proof}
  \begin{eqnarray*}
  &&\sem{\promote(\setgraph{\vect{x \leftarrow N}}{M})}\rho (\vect{u},v) 
\\
&=&
\chi(\sem{\setgraph{\vect{x \leftarrow N}}{M}}\rho (\vect{u},v))
\\
&=&
\chi(\sem{\vect{N}}\rho \vect{u} \wedge \sem{M}\rho[ \vect{x\mapsto u}] v)
\\
&=&
\chi(\sem{\vect{N}}\rho \vect{u}) \times \chi(\sem{M}\rho[ \vect{x\mapsto u}] v)
\\
&=&
\chi(\sem{\vect{N}}\rho \vect{u}) \times \sem{\promote(M)}\rho[
    \vect{x\mapsto u}] v
\\
&=&
 \sem{\baggraph{\vect{x \leftarrow
     N}}{\promote(M)}}\rho v
\end{eqnarray*}
\qed\end{proof}
\begin{corollary}\label{cor:graph-iota}
  $\promote(\setgraph{x \leftarrow N}{M})\gapp O \equiv \promote(M[O/x])
  ~\bagwhere~ O \in N$
\end{corollary}

\begin{lemma}\label{lem:graph-delta}
  $\distinct(\baggraph{\vect{x \leftarrow N}}{M}) \equiv \setgraph{\vect{x \leftarrow N}}{\distinct(M)}$
\end{lemma}
\begin{proof}
  \begin{eqnarray*}
  &&\sem{\distinct(\baggraph{\vect{x \leftarrow N}}{M})}\rho (\vect{u},v) 
\\
&=&
\zeta(\sem{\baggraph{\vect{x \leftarrow N}}{M}}\rho (\vect{u},v))
\\
&=&
\zeta(\chi(\sem{\vect{N}}\rho) \vect{u} \times \sem{M}\rho[ \vect{x\mapsto u}] v)
\\
&=&
\zeta(\chi(\sem{\vect{N}}\rho \vect{u})) \wedge \zeta(\sem{M}\rho[ \vect{x\mapsto u}] v)
\\
&=&
\sem{\vect{N}}\rho \vect{u} \wedge \sem{\distinct(M)}\rho[
    \vect{x\mapsto u}] v
\\
&=&
 \sem{\baggraph{\vect{x \leftarrow
     N}}{\distinct(M)}}\rho v
\end{eqnarray*}
\qed\end{proof}

\begin{lemma}\label{lem:graph-union}
  $\graph{\vect{x \leftarrow N}}{M_1} \cup\graph{\vect{x \leftarrow N}}{M_2}
  \equiv \graph{\vect{x\leftarrow N}}{M_1 \cup M_2}$
\end{lemma}
\begin{proof}
  \begin{eqnarray*}
&&\sem{\graph{\vect{x \leftarrow N}}{M_1} \cup\graph{\vect{x \leftarrow N}}{M_2}}\rho (\vect{u},v)\\
&=&
\sem{\graph{\vect{x \leftarrow N}}{M_1}}\rho(\vect{u},v) \cup\sem{\graph{\vect{x \leftarrow
    N}}{M_2}}\rho(\vect{u},v)
\\
&=&
\vect{\sem{N}}\rho \vect{u} \wedge \sem{M_1}\rho[\vect{x\mapsto u}] v \vee
\vect{\sem{N}}\rho \vect{u} \wedge \sem{M_2}\rho[\vect{x\mapsto u}] v
\\
&=&
\vect{\sem{N}}\rho \vect{u})  \wedge (\sem{M_1}\rho[\vect{x\mapsto u}]  v \vee \sem{M_2}\rho[\vect{x\mapsto u}] v)
\\
&=&
\vect{\sem{N}}\rho \vect{u})  \wedge (\sem{M_1\cup M_2}\rho[\vect{x\mapsto u}]  v)
\\
&=& \sem{ \graph{\vect{x\leftarrow N}}{M_1 \cup M_2}}\rho (\vect{u},v)
\end{eqnarray*}
 \qed\end{proof}

\begin{lemma}
  $\graph{\vect{x \leftarrow N}}{M_1} \uplus \graph{\vect{x \leftarrow N}}{M_2}
  \equiv \graph{\vect{x\leftarrow N}}{M_1  \uplus M_2}$
\end{lemma}
\begin{proof}
  \begin{eqnarray*}
&&\sem{\graph{\vect{x \leftarrow N}}{M_1} \uplus \graph{\vect{x \leftarrow N}}{M_2}}\rho (\vect{u},v)\\
&=&
\sem{\graph{\vect{x \leftarrow N}}{M_1}}\rho(u,v) \uplus\sem{\graph{\vect{x \leftarrow
    N}}{M_2}}\rho(\vect{u},v)
\\
&=&
\chi(\vect{\sem{N}}\rho \vect{u}) \times \sem{M_1}\rho[\vect{x\mapsto u}] v +
\chi(\vect{\sem{N}}\rho \vect{u}) \times \sem{M_2}\rho[\vect{x\mapsto u}] v
\\
&=&
\chi(\vect{\sem{N}}\rho \vect{u})  \times
    (\sem{M_1}\rho[\vect{x\mapsto u}]  v + \sem{M_2}\rho[\vect{x\mapsto u}] v)
\\
&=&
\chi(\vect{\sem{N}}\rho \vect{u})  \times (\sem{M_1 \uplus
M_2}\rho[\vect{x\mapsto u}]  v)
\\
&=& \sem{ \graph{\vect{x\leftarrow N}}{M_1 \uplus M_2}}\rho (\vect{u},v)
\end{eqnarray*}
\qed\end{proof}

\begin{lemma}\label{lem:graph-diff}
  $\graph{\vect{x \leftarrow N}}{M_1} -\graph{\vect{x \leftarrow N}}{M_2}
  \equiv \graph{\vect{x\leftarrow N}}{M_1 - M_2}$
\end{lemma}
\begin{proof}
  \begin{eqnarray*}
&&\sem{\graph{\vect{x \leftarrow N}}{M_1} -\graph{\vect{x \leftarrow N}}{M_2}}\rho (\vect{u},v)\\
&=&
\sem{\graph{\vect{x \leftarrow N}}{M_1}}\rho(\vect{u},v) -\sem{\graph{\vect{x \leftarrow
    N}}{M_2}}\rho(\vect{u},v)
\\
&=&
\chi(\vect{\sem{N}}\rho \vect{u}) \times \sem{M_1}\rho[\vect{x\mapsto u}] v -
\chi(\vect{\sem{N}}\rho \vect{u}) \times \sem{M_2}\rho[\vect{x\mapsto u}] v
\\
&=&
\chi(\vect{\sem{N}}\rho \vect{u})  \times (\sem{M_1}\rho[\vect{x\mapsto u}]  v - \sem{M_2}\rho[\vect{x\mapsto u}] v)
\\
&=&
\chi(\vect{\sem{N}}\rho \vect{u})  \times (\sem{M_1-M_2}\rho[\vect{x\mapsto u}]  v)
\\
&=& \sem{ \graph{\vect{x\leftarrow N}}{M_1 - M_2}}\rho (\vect{u},v)
\end{eqnarray*}
 \qed\end{proof}

\begin{corollary}\label{cor:graph-diff}
  $(\graph{x \leftarrow N}{M_1} -\graph{x \leftarrow N}{M_2})\gapp O
  \equiv (M_1[O/x] - M_2[O/x])~ \bagwhere~ O \in N$
\end{corollary}

We can now use these equivalences to show the correctness of the
delateralization rules for promotion and difference:
\begin{theorem}
\[  \biguplus\msetlit{M \mid x \leftarrow N, y \leftarrow \promote(P)} \equiv
\biguplus\msetlit{M \mid x \leftarrow N, y \leftarrow
  \promote(\graph{x \leftarrow N}{P}) \gapp x}\]
\end{theorem}
\begin{proof}
We use Cor.~\ref{cor:graph-iota} and Lemma~\ref{lem:graph-elem} and
standard equivalences for monadic comprehensions:
\begin{eqnarray*}
&&\biguplus\msetlit{M \mid x \leftarrow N, y \leftarrow \promote(P)}
\\
&\equiv &
\biguplus\msetlit{\biguplus\msetlit{M \mid y \leftarrow
  \promote(P) }\mid x \leftarrow N}\\
&\equiv &
\biguplus\msetlit{\biguplus\msetlit{M \mid y \leftarrow
  \promote(P) }~\bagwhere ~x \in N\mid x \leftarrow N}\\
&\equiv &
\biguplus\msetlit{\biguplus\msetlit{M \mid y \leftarrow
  \promote(P) ~\bagwhere ~x \in N}\mid x \leftarrow N}\\
&\equiv &
\biguplus\msetlit{M \mid x \leftarrow N, y \leftarrow
  \promote(P) ~\bagwhere ~x \in N}\\
&\equiv &
\biguplus\msetlit{M \mid x \leftarrow N, y \leftarrow
  \promote(\graph{x \leftarrow N}{P}) \gapp x}
\end{eqnarray*}
\qed\end{proof}

\begin{theorem}
  \[
\biguplus\msetlit{M \mid x \leftarrow N, y \leftarrow P_1-P_2} \equiv
\biguplus\msetlit{M \mid x \leftarrow N, y \leftarrow
\graph{x \leftarrow \distinct(N)}{P_1} - \graph{x \leftarrow \distinct(N)}{P_2}}
\]
\end{theorem}
\begin{proof}
We use Cor.~\ref{cor:graph-diff} and Lemma~\ref{lem:graph-elem} and
standard equivalences for monadic comprehensions:
\begin{eqnarray*}
&&
\biguplus\msetlit{M \mid x \leftarrow N, y \leftarrow P_1-P_2} \\
&\equiv&
\biguplus\msetlit{\biguplus\msetlit{M \mid  y \leftarrow P_1-P_2} \mid
  x \leftarrow N}\\
&\equiv&
\biguplus\msetlit{\biguplus\msetlit{M \mid  y \leftarrow P_1-P_2}
         ~\bagwhere~ x\in \distinct(N)\mid
  x \leftarrow N}\\
&\equiv&
\biguplus\msetlit{\biguplus\msetlit{M \mid  y \leftarrow (P_1-P_2) ~\bagwhere~ x\in \distinct(N)}
         \mid
  x \leftarrow N}\\
&\equiv&
\biguplus\msetlit{M \mid x \leftarrow N, y \leftarrow
(\graph{x \leftarrow \distinct{N}}{P_1} - \graph{x \leftarrow
         \distinct{N}}{P_2}) \gapp x}
\end{eqnarray*}
\qed\end{proof}

\begin{theorem}
  \[
\bigcup\setlit{M \mid x \leftarrow N, y \leftarrow \distinct(P_1-P_2)} \equiv
\bigcup\setlit{M \mid x \leftarrow N, y \leftarrow
\distinct(\graph{x \leftarrow N}{P_1} - \graph{x \leftarrow N}{P_2})}
\]
\end{theorem}
\begin{proof}
We use Lemmas~\ref{lem:graph-elem}, \ref{lem:graph-delta} and \ref{lem:graph-diff} and
standard equivalences for monadic comprehensions:
\begin{eqnarray*}
&&
\bigcup\setlit{M \mid x \leftarrow N, y \leftarrow \distinct(P_1-P_2)} \\
&\equiv&
\bigcup\setlit{\bigcup\setlit{M \mid  y \leftarrow \distinct(P_1-P_2)} \mid
  x \leftarrow N}\\
&\equiv&
\bigcup\setlit{\bigcup\setlit{M \mid  y \leftarrow \distinct(P_1-P_2)}
         ~\setwhere~ x\in N\mid
  x \leftarrow N}\\
&\equiv&
\bigcup\setlit{\bigcup\setlit{M \mid  y \leftarrow \distinct(P_1-P_2) ~\setwhere~ x\in N}
         \mid
  x \leftarrow N}\\
&\equiv&
\bigcup\setlit{\bigcup\setlit{M \mid  y \leftarrow \graph{x \leftarrow
         N}{\distinct(P_1-P_2)} \gapp x}
         \mid
  x \leftarrow N}\\
&\equiv&
\bigcup\setlit{\bigcup\setlit{M \mid  y \leftarrow \distinct(\graph{x \leftarrow
         N}{P_1-P_2}) \gapp x}
         \mid
  x \leftarrow N}\\
&\equiv&
\bigcup\setlit{M \mid x \leftarrow N, y \leftarrow
\distinct(\graph{x \leftarrow N}{P_1} - \graph{x \leftarrow
         N}{P_2}) \gapp x}
\end{eqnarray*}
\qed\end{proof}

Now, to prove that delateralization eventually terminates, we consider
a metric on query expressions defined as follows: given an expression
in normal form, for each subexpression of the form $\promote(N)$ or
$M - N$, add up the number of free variables occurring in $M,N$.
\begin{eqnarray*}
\|\emptymset\| &=& 0\\
\|\setlit{M}\| = \msetlit{M}\} &=& \|M\|\\
\|M\cup N\| = \|M\uplus N\| &=& \|M\|+\|N\|\\
\|M - N \| &=& \|M\| + \|N\| + |FV(M,N)|\\
\|\iota(M) \| &=& \|M\|+ |FV(M,N)|\\
\|\delta(M) \| &=& \|M\|\\
\|\bigcup\setlit{M \mid x \leftarrow N}\| = \|\biguplus\msetlit{M \mid x \leftarrow N} &=& \|M\| + \|N\|\\
\|M \ \setwhere \ N\| = \|M \ \bagwhere \ N\| &=& \|M \| + \|N\| \\
\|M\| &=& 0 \qquad \text{otherwise}
\end{eqnarray*}

If the metric is zero,
then the query is fully delateralized.  Combining the basic
delateralization steps above with commutativity, any expression with
nonzero metric can be rewritten so as to decrease the metric (though
possibly increasing the query size).  We can also undo the effects of
commutativity steps to restore the original order of generators, to
preserve the query structure as much as possible for readability.

\begin{theorem}
  Given $M$ with $\|M\| > 0$, there exists an equivalent $M'$ with
  $\|M'\| < \|M\|$ that can be obtained by applying commutativity and
  basic rewrites.  Hence, there exists an equivalent
  fully-delateralized $M''$ with $\|M''\| = 0$.
\end{theorem}
\begin{proof}
  The proof requires establishing that whenever $\|M\| > 0$, there
  exists at least one outermost subexpression $M_0$ of the form
  $\promote(N)$ or $N-P$ with $\|M_0\| > 0$.  That is, $M_0$ should
  not be a subexpression of any larger such subexpression of $M$
  having the same property.  Moreover, $M_0$ must occur as a
  generator.  We need to show that $M_0$ therefore contains at least
  one free record variable bound earlier in the same comprehension.
  We can show this by inspection of normal forms.  Since this is the
  case, then (if the generator is not already adjacent) we can commute
  it to be adjacent to $M_0$ and then apply one of the
  delateralization rules, decreasing $\|M_0\|$ and hence $\|M\|$.
\qed\end{proof}


\section{Proofs for Section~\ref{sec:shredding}}

\begin{lemma}\label{lem:shredmono}
If $\Phi;\Theta \vdash M \shreds \flat{M} \mid \Psi$, then $\Psi \supseteq \Phi$.
\end{lemma}

\begin{lemma}\label{lem:substweak}
Let $M$ an \NRCg term and $\Phi$ a shredding set. If $\FV(M) \subseteq \dom(\Phi)$, then for all $\Phi' \supseteq \Phi$ 
we have $M\Phi = M\Phi'$.

Furthermore, let $\Xi$ and $\Xi'$ be the shredding value sets corresponding to $\Phi$ and $\Phi'$: then $\stitch{\Xi}{\defun{M}} = \stitch{\Xi'}{\defun{M}}$.
\end{lemma}

In the following proof, whenever $\Theta = \vect{x \gets F}$, we use the abbreviation:
\[
  \sem{\Theta}\rho~\vect{v} = \bigwedge_i \sem{F_i}\rho[x_1 \mapsto v_1, \ldots, x_{i-1} \mapsto v_{i-1}]~v_i
\]

\subsubsection{Theorem~\ref{thm:stitchsound}.}
\emph{
  Let $\Theta$ be well-typed and $ty(\Theta) \vdash M : \sigma$.
  Let $\Phi$ be well-typed, and  suppose $\Phi;\Theta \vdash M \shreds \flat{M} \orelse \Psi$. Let $\Xi$ be the result of evaluating
  the flattened queries in $\Psi$ as above.  Then
  $\sem{\flat{M}\Psi}\rho = \sem{\stitch{\Xi}{\defun{\flat{M}}:\tau}}\rho$.
}
\begin{proof}
  We proceed by induction on the shredding judgment. We comment the two key cases:

  \begin{itemize}
  \item For set comprehension: 
  \[
    \begin{array}{rl}
    \Phi;\Theta \vdash 
    & \comprehension{\setlit{M}~\kwwhere~X \mid \Theta'} \shreds \varphi \gapp(\dom(\Theta))
    \\
    \mid 
    & \Psi[\varphi \mapsto \graph{\Theta}{\comprehension{\setlit{\flat{M}}~\kwwhere~X \mid \Theta'}}]
    \end{array}
  \]
  where we wrote $\Theta' = \vect{x \gets F}$ for conciseness. Let $\Xi$ be the shredding value set for $\Psi$, and $\Xi'$ the shredding value set for $\Psi[\varphi \mapsto \graph{\Theta}{\comprehension{\setlit{\flat{M}}~\kwwhere~X \mid \Theta'}}]$. We rewrite the rhs:
  \[
  \begin{array}{l}
    \sem{\stitch{\Xi'}{\kwindex(\varphi,\dom(\Theta))}}\rho~u
    \\

    = \sem{\comprehension{\setlit{\stitch{\Xi'}{p.2}}~\kwwhere~p.1 = \tuple{\dom(\Theta)} \mid p \gets \Xi'(\varphi)}}\rho~u
    \\

    = \sem{\bigcup\left\{
      \begin{array}{l}
        \setlit{\stitch{\Xi'}{p.2}}~\kwwhere~p.1 = \tuple{\dom(\Theta)} 
        \\
        \qquad 
        \mid p \gets \cS(\norm(\defun{\graph{\Theta^*}{\comprehension{\setlit{\flat{M}^*}~\kwwhere~X^* \mid \Theta'^*}}}))
      \end{array}
    \right\}} \rho~u
    \\

    = \sem{\bigcup\left\{
      \begin{array}{l}
        \setlit{\stitch{\Xi'}{p.2}}~\kwwhere~p.1 = \tuple{\dom(\Theta)} 
        \\
        \qquad 
        \mid p \gets \comprehension{\setlit{\tuple{\tuple{\dom(\Theta^*)},\defun{\flat{M}^*}}}~\kwwhere~X^* \mid \Theta^*, \Theta'^*}
      \end{array}
    \right\}} \rho~u
    \\

    = \sem{\comprehension{
        \setlit{\stitch{\Xi'}{\defun{\flat{M}^*}}}~\kwwhere~(\tuple{\dom(\Theta^*)} = \tuple{\dom(\Theta)} \land X^*) \mid \Theta^*, \Theta'^*
     }}\rho~u
    \\

    = \bigvee_{\vect{v},\vect{v'}} (\sem{\stitch{\Xi}{\defun{\flat{M}^*}}}\rho' = u) \land \sem{\dom(\Theta^*) = \dom(\Theta)}\rho' \land \sem{X^*}\rho' \land \sem{\Theta^*,\Theta'^*}\rho~\vect{v},\vect{v'}
    \\

    = \bigvee_{\vect{v'}} (\sem{\stitch{\Xi}{\defun{\flat{M}}}}\rho'' = u) \land \sem{X}\rho'' \land \sem{\Theta'}\rho~\vect{v'}
  \end{array}
  \]
  where we alpha-renamed $\dom(\Theta)$ to a fresh $\dom(\Theta^*)$ within $\flat{M}$, $X$ and $\Theta'$ (yielding $\flat{M}^*$, $X^*$, and $\Theta'^*$, and we have set $\rho' = \rho[\dom(\Theta^*) \mapsto \vect{v}, \dom(\Theta') \mapsto \vect{v'}]$, $\rho'' = \rho[\dom(\Theta') \mapsto \vect{v'}]$. Note that the renaming involving $\Theta^*$ is undone in the last step through the evaluation of $\sem{\dom(\Theta^*) = \dom(\Theta)}\rho'$.

  We then rewrite the lhs:
  \[
    \begin{array}{l}
      \sem{(\varphi \gapp \dom(\Theta))\Psi'}\rho~u
      \\
      = \sem{\graph{\Theta}{\comprehension{\setlit{\flat{M}\Psi'}~\kwwhere~X \mid \Theta'}}\gapp(\dom(\Theta))}\rho~u
      \\
      = \sem{\comprehension{\setlit{\flat{M}\Psi'}~\kwwhere~X \mid \Theta'}}\rho~u
      \\
      = \bigvee_{\vect{v'}} (\sem{\flat{M}\Psi'}\rho'' = u) \land \sem{X}\rho'' \land \sem{\Theta'}\rho~\vect{v'}
    \end{array}  
  \]

  By Lemma~\ref{lem:substweak} and by induction hypotheses, we prove:
  \[
  \sem{\flat{M}\Psi'}\rho'' = \sem{\flat{M}\Psi}\rho'' = \sem{\stitch{\Xi}{\defun{\flat{M}}}}\rho'' = \sem{\stitch{\Xi'}{\defun{\flat{M}}}}\rho''
  \]
  which we combine with the previous calculations to prove the thesis.

  \item For set union: 
    \begin{prooftree}
    \AxiomC{$\varphi \notin \dom(\Phi_n)$}
    \noLine
    \UnaryInfC{$(\Phi_{i-1}; \Theta \vdash C_i \shreds \psi_i \gapp \dom(\Theta) \orelse \Phi_i)_{i=1,\ldots,n}$}
    \UnaryInfC{$
      \begin{array}{rl}
        \Phi_0; \Theta \vdash & \bigcup \vect{C} \shreds \varphi \gapp \dom(\Theta)
        \\
        \orelse & (\Phi_n \setminus \vect{\psi})[\varphi \mapsto \bigcup\vect{\Phi_n(\psi)}]
      \end{array}$}
    \end{prooftree}
    Let $\Xi_i$ be the shredding value set for each $\Phi_i$, and $\Xi'$ the shredding value set for $\Phi' := (\Phi_n \setminus \vect{\psi})[\varphi \mapsto \bigcup\vect{\Phi_n(\psi)}]$. We rewrite the rhs:
  \[
  \begin{array}{l}
    \sem{\stitch{\Xi'}{\kwindex(\varphi,\dom(\Theta))}}\rho~u
    \\

    = \sem{\comprehension{\setlit{\stitch{\Xi'}{p.2}}~\kwwhere~p.1 = \tuple{\dom(\Theta)} \mid p \gets \Xi'(\varphi)}}\rho~u
    \\

    = \sem{\comprehension{\setlit{\stitch{\Xi'}{p.2}}~\kwwhere~p.1 = \tuple{\dom(\Theta)}
        \mid p \gets \cS(\norm(\defun{\bigcup\vect{\Phi_n(\psi)}}))}}\rho~u
    \\

    = \sem{\bigcup_i \comprehension{\setlit{\stitch{\Xi'}{p.2}}~\kwwhere~p.1 = \tuple{\dom(\Theta)}
        \mid p \gets \cS(\norm(\Phi_n(\psi_i)}} \rho~u
    \\

    = \sem{\bigcup_i \comprehension{\setlit{\stitch{\Xi'}{p.2}}~\kwwhere~p.1 = \tuple{\dom(\Theta)}
        \mid p \gets \cS(\norm(\Phi_i(\psi_i)}} \rho~u
    \\

    = \bigvee_i \sem{\stitch{\Xi_i}{\kwindex(\psi_i,\dom(\Theta))}}\rho~u
  \end{array}
  \]
  We also rewrite the lhs:
  \[
    \begin{array}{l}
      \sem{(\varphi \gapp \dom(\Theta))\Phi'}\rho~u
      \\
      = \sem{(\bigcup \vect{(\Phi_n(\psi)\Phi')})\gapp(\dom(\Theta))}\rho~u
      \\
      = \sem{\bigcup \vect{(\Phi_n(\psi)\Phi')}}\rho~(\sem{\dom(\Theta)}\rho~u)
      \\
      = \bigvee_i \sem{(\Phi_n(\psi_i)\Phi')}\rho~(\sem{\dom(\Theta)}\rho~u)
      \\
      = \bigvee_i \sem{(\Phi_i(\psi_i)\Phi_i)}\rho~(\sem{\dom(\Theta)}\rho~u)
      \\
      = \bigvee_i \sem{\psi_i\Phi_i}\rho~(\sem{\dom(\Theta)}\rho~u)
      \\
      = \bigvee_i \sem{(\psi_i \gapp(\dom(\Theta)))\Phi_i}\rho~u)
    \end{array}  
  \]

  By induction hypothesis, we prove:
  \[
    \sem{(\psi_i \gapp(\dom(\Theta)))\Phi_i}\rho
    = \sem{\stitch{\Xi_i}{\defun{\psi_i \gapp(\dom(\Theta))}}}\rho
    = \sem{\stitch{\Xi_i}{\kwindex(\psi_i,\dom(\Theta))}}\rho
  \]
  which we combine with the previous calculations to prove the thesis.
  \qed
  \end{itemize}
\end{proof}

\end{techreport}


\vfill

{\small\medskip\noindent{\bf Open Access} This chapter is licensed under the terms of the Creative Commons\break Attribution 4.0 International License (\url{http://creativecommons.org/licenses/by/4.0/}), which permits use, sharing, adaptation, distribution and reproduction in any medium or format, as long as you give appropriate credit to the original author(s) and the source, provide a link to the Creative Commons license and indicate if changes were made.}

{\small \spaceskip .28em plus .1em minus .1em The images or other third party material in this chapter are included in the chapter's Creative Commons license, unless indicated otherwise in a credit line to the material.~If material is not included in the chapter's Creative Commons license and your intended\break use is not permitted by statutory regulation or exceeds the permitted use, you will need to obtain permission directly from the copyright holder.}

\medskip\noindent\includegraphics{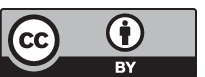}

\end{document}